\documentclass[journal]{IEEEtran}
\makeatletter
\def\ps@headings{%
\def\@oddhead{\mbox{}\scriptsize\rightmark \hfil \thepage}%
\def\@evenhead{\scriptsize\thepage \hfil \leftmark\mbox{}}%
\def\@oddfoot{}%
\def\@evenfoot{}}
\makeatother
\usepackage{amsmath, amssymb, amsthm, amsfonts, amsxtra, latexsym, amscd}
\usepackage{amssymb}
\usepackage{latexsym}
\usepackage{amsthm}
\usepackage{graphicx}
\usepackage{verbatim}
\usepackage{mathrsfs}
\usepackage{algorithmic}
\usepackage{float}
\usepackage[lofdepth, lotdepth]{subfig}
\usepackage{url}
\usepackage{array}

\usepackage[usenames,dvipsnames]{pstricks}
\usepackage{epsfig}
\usepackage{pst-grad} 

\theoremstyle{plain}
  \newtheorem{theorem}{Theorem}[section]
  \newtheorem{lemma}[theorem]{Lemma}
  \newtheorem{proposition}[theorem]{Proposition}
  \newtheorem{corollary}[theorem]{Corollary}
\theoremstyle{definition}
  \newtheorem{definition}[theorem]{Definition}
  
  \newtheorem{example}[theorem]{Example}
  \newtheorem{remark}[theorem]{Remark}

\newcommand{\mrt}{\mathsf{minrk}_2}
\newcommand{\mrq}{\mathsf{minrk}_q}

\newcommand{\A}{{\mathcal A}}

\newcommand{\C}{{\mathcal C}}

\newcommand{\D}{{\mathcal D}}
\newcommand{\E}{{\mathcal E}}
\newcommand{\F}{{\mathcal F}}
\newcommand{\G}{{\mathcal G}}

\newcommand{\LL}{{\mathcal L}}
\newcommand{\Q}{{\mathcal Q}}
\newcommand{\X}{{\mathcal X}}

\newcommand{\V}{{\mathcal V}}

\newcommand{\N}{{\mathcal N}}

\newcommand{\bM}{{\boldsymbol{M}}} 
 
\newcommand{\bJ}{{\boldsymbol J}}

\newcommand{\bB}{{\boldsymbol B}}

\newcommand{\supp}{{\sf supp}}

\newcommand{\rank}{{\mathsf{rank}_q}}
\newcommand{\rankt}{{\mathsf{rank}_2}}

\newcommand{\define}{\stackrel{\mbox{\tiny $\triangle$}}{=}}

\newcommand{\bu}{{\boldsymbol u}}

\newcommand{\bO}{{\boldsymbol 0}}

\newcommand{\bx}{{\boldsymbol{x}}}

\newcommand{\be}{{\boldsymbol e}}

\newcommand{\aG}{{\alpha(\G)}}

\newcommand{\mG}{\mathsf{mm}(\G)}
\newcommand{\cc}{{\sf{cc}}}
\newcommand{\cG}{{\sf{cc}(\G)}}
\newcommand{\Gc}{\overline{{\mathcal{G}}}}
\newcommand{\VG}{\mathcal{V}(\mathcal{G})}

\newcommand{\EG}{\mathcal{E}(\mathcal{G})}

\newcommand{\aD}{{\alpha(\D)}}

\newcommand{\cD}{{\sf{cc}(\D)}}
\newcommand{\Dc}{\overline{{\mathcal{D}}}}
\newcommand{\VD}{\mathcal{V}(\mathcal{D})}
\newcommand{\VDc}{\mathcal{V}(\overline{\mathcal{D}})}
\newcommand{\ED}{\mathcal{E}(\mathcal{D})}
\newcommand{\EDc}{\mathcal{E}(\overline{\mathcal{D}})}

\newcommand{\nD}{\nu_0(\mathcal{D})}
\newcommand{\tD}{\tau_0(\mathcal{D})}
\newcommand{\neD}{\nu_1(\mathcal{D})}
\newcommand{\teD}{\tau_1(\mathcal{D})}
\newcommand{\LD}{{\mathcal{L}}({\mathcal{D}})}

\newcommand{\al}{\alpha}

\newcommand{\bt}{\beta}

\newcommand{\kp}{\kappa}

\newcommand{\nin}{\noindent}

\newcommand{\ra}{\rightarrow}

\newcommand{\ff}{\mathbb{F}}
\newcommand{\ft}{\mathbb{F}_2}
\newcommand{\fq}{\mathbb{F}_q}

\newcommand{\et}{{\emph{et al.}}}

\newcommand{\fkE}{{\mathfrak E}}

\title{Optimal Index Codes with Near-Extreme Rates
\thanks{%
This work was supported in part by the National Research Foundation of Singapore (Research Grant NRF-CRP2-2007-03).
This paper was presented in part at the IEEE International Symposium on Information Theory, Cambridge, MA, USA, July 2012.
}
\thanks{%
S. H. Dau was with the Division of Mathematical Sciences, School of Physical and Mathematical Sciences, Nanyang Technological University, 21~Nanyang Link, Singapore 637371. He is now with the SUTD-MIT International Design Centre, Singapore University of Technology and Design, 20 Dover Drive, Singapore 138682
(e-mail: dausonhoang84@gmail.com).
}
\thanks{%
    Y. M. Chee is with the Division of Mathematical Sciences, School of Physical and Mathematical Sciences,
    Nanyang Technological University, 21~Nanyang Link, Singapore 637371  
    (e-mail: ymchee@ntu.edu.sg). 
}
\thanks{%
V. Skachek was with the Coordinated Science Laboratory, University of Illinois at Urbana-Champaign,
    1308 W. Main Street, Urbana, IL 61801, USA. He is now with the Institute of Computer Science, Faculty of Mathematics and 
		Computer Science, University of Tartu, J. Liivi 2-216, Tartu 50409, Estonia (e-mail: vitaly.skachek@ut.ee).			
}
} 

\author{Son Hoang Dau, Vitaly Skachek, and Yeow Meng Chee, \emph{Senior Member}, \emph{IEEE}}

\begin{document}
\maketitle

\begin{abstract}
The \emph{min-rank} of a \emph{digraph} was shown by Bar-Yossef~{\et}~(2006) to represent the length of an optimal scalar linear solution of the corresponding instance of the Index Coding with Side Information (ICSI) problem. 
In this work, the graphs and digraphs of near-extreme min-ranks are characterized. 
Those graphs and digraphs correspond to the ICSI instances having near-extreme transmission rates when using optimal scalar linear index codes. In particular, it is shown that the decision problem whether a digraph has min-rank two is NP-complete. By contrast, the same question for \emph{graphs} can be answered in polynomial time.  

Additionally, a new upper bound on the min-rank of a digraph, the \emph{circuit-packing bound}, is presented. This 
bound is often tighter than the previously known bounds. By employing this new bound, we present several families of digraphs whose min-ranks can be found in polynomial time. 
\end{abstract}

\section{Introduction}

Building communication schemes which allow participants to communicate efficiently
has always been a challenging yet intriguing problem for information theorists. 
Index Coding with Side Information (ICSI) (\cite{BirkKol98, BirkKol2006}) is a 
communication scheme dealing with broadcast channels in which receivers have prior side information 
about the messages to be transmitted. By using coding and exploiting the knowledge about the side information,
the sender may significantly reduce the number of required transmissions compared
with the straightforward approach. As a consequence, the efficiency of communication over 
this type of broadcast channels could be dramatically improved. 
Apart from being a special case of the well-known (non-multicast) Network Coding problem
(\cite{Ahlswede, KoetterMedard2003}), the ICSI
problem has also found various potential applications on its owns, such as audio- and 
video-on-demand, daily newspaper delivery, data pushing, and opportunistic wireless networks
(\cite{BirkKol98, BirkKol2006,Yossef, Rouayheb2009, Katti2006, Katti2008}). 

In the work of Bar-Yossef {\et} \cite{Yossef}, the optimal transmission rate of scalar linear index codes for an ICSI instance
was neatly characterized by the so-called \emph{min-rank} of the side information digraph (i.e., directed graph, see Section~\ref{sec:not_def} for definitions) corresponding to that instance. 
The concept of min-rank of a graph (i.e., undirected graph, see Section~\ref{sec:not_def} for definitions) goes back to Haemers \cite{Haemers1978}. Min-rank serves as an upper bound for the celebrated Shannon capacity of a graph \cite{Shannon1956}. This upper bound, as pointed out by Haemers, 
although is usually not as good as the Lov\'asz bound \cite{Lovasz1979}, is sometimes tighter and easier to compute.
It was shown by Peeters \cite{Peeters96} that computing the min-rank of a general graph 
(that is, the Min-Rank problem) is a hard task. 
More specifically, Peeters showed that deciding whether the min-rank of a graph 
is smaller than or equal to three is an NP-complete problem.
 
The work of Bar-Yossef {\et}~\cite{Yossef} has stimulated
the interest in the Min-Rank problem. 
Exact and heuristic algorithms for finding min-ranks over the binary field of digraphs were developed in the work of Chaudhry and Sprintson \cite{ChaudhrySprintson}. The min-ranks of random digraphs are investigated by Haviv and Langberg \cite{HavivLangberg2011}. A dynamic programming approach was proposed by Berliner and Langberg~\cite{BerlinerLangberg2011} to compute min-ranks of outerplanar graphs in polynomial time. 
Algorithms to approximate min-ranks of graphs with bounded
min-ranks were studied by Chlamtac and Haviv~\cite{ChlamtacHaviv2011}.

In this paper, we study graphs and digraphs that have near-extreme min-ranks. 
In other words, we study ICSI instances with $n$ receivers for which optimal 
\emph{scalar linear} index codes
have transmission rates $2$, $n-2$, $n-1$, or $n$.
In particular, we show that deciding whether a digraph has min-rank two over the \emph{binary} field is an NP-complete problem. By contrast, a graph has min-rank two over any finite field if and only it is not a complete graph and its complement is bipartite, a condition which can be verified in polynomial time (see, for instance, West~\cite[p. 495]{West}). Very recently, it was found by Maleki {\et}~\cite{Maleki2012} that the same problem for digraph over sufficiently large field can be solved in polynomial time.

The characterizations of graphs and digraphs with near-extreme min-ranks are summarized in the table below. 
The star mark ``$*$'' indicates that the result is established in
this paper. The dagger mark ``$\dagger$'' indicates that 
the result is proved only for the binary field. 
\begin{table}[H]
	\centering
		\begin{tabular}{|m{1.16cm}|m{3.85cm}|m{2.52cm}|}
			\hline \vspace{2pt}
			Min-Rank & Graph $\G$ & Digraph $\D$\\
			\hline
			$1$ & $\G$ is complete (trivial) & $\D$ is complete (trivial) \\
			\hline \vspace{2pt}
			$2$ & $\G$ is not complete and $\Gc$ is $2$-colorable
			(\cite{Peeters96})
			& $\D$ is not complete and $\Dc$ is fairly $3$-colorable$^{*\dagger}$\\
			\hline \vspace{2pt}
			$n-2$ & $\G$ (connected) has a maximum matching of size two and does not
			contain $F$ (Fig.~\ref{fig:forbidden_subgraph}) as a subgraph$^*$ & unknown \\
			\hline \vspace{2pt}
			$n - 1$ & $\G$ (connected) is a star graph$^*$ & unknown\\
			\hline
			$n$ & $\G$ has no edges (trivial) & $\D$ has no circuits$^*$ \\
			\hline
		\end{tabular}
\end{table}
The near-extreme cases are of significant interest from both theoretical and practical
points of view. 
It is known that the Min-Rank Problem is NP-hard \cite{Peeters96} ($\mrq(\G) = 3$ is hard to verify). \emph{Theoretically}, it is desirable to further understand, which values of the min-rank in the range between $1$ and $n$ are still easy to verify, and for which values it is hard. It turns out that for graphs and digraphs, the easy-hard turning points are different. For graphs, the turning points are $3$ and some value smaller than $n-2$ (not exactly known). By contrast, for digraphs, the easy-hard turning points are $2$ (proved in this work) and $n-1$ (conjectured).    
\emph{Practically}, the use of length-one index codes in wireless communications has already been proposed (for instance, see COPE~\cite{Katti2006},~\cite{Traskov2009},~\cite{ZhanXu2010}), due to their simplicity and efficiency. However, the variety of scenarios where an index code of length one is applicable is limited (each client must know all except one message). An index
code of length two is obviously the next potential candidate to be used.  

In this paper, we also introduce a new upper bound for the min-rank of a digraph, namely the circuit-packing bound, which, in certain cases, is far tighter than the clique-cover bound. This upper bound was first presented by Chaudhry {\et}~\cite{ChaudhryAsadSprintsonLangberg2011}, and was found independently by the authors of this paper approximately 
at the same time.

So far, families of graphs and digraphs whose min-ranks are either known or computable in polynomial time are the followings. For \emph{graphs}, they are odd holes and odd anti-holes~\cite{Yossef-journal}, perfect graphs~\cite{Yossef-journal}, and outerplanar graphs~\cite{BerlinerLangberg2011}. For \emph{digraphs}, they are acyclic digraphs~\cite{Yossef-journal}. In this work, we point out several new families of \emph{digraphs} for which the circuit-packing bound is tight. For such families of digraphs, min-ranks can be found in polynomial time. 

In the context of index coding, we only study min-ranks of digraphs over a \emph{finite} field $\fq$. However, all of our results, except Theorem~\ref{thm:mr_color}, Corollary~\ref{coro:mr_color}, and Theorem~\ref{thm:NP_complete}, still hold for an \emph{arbitrary} field $\ff$. This is because the characteristic of the field does not play any role in their proofs. 

The paper is organized as follows. 
Basic notation and definitions are presented in Section~\ref{sec:not_def}. 
The ICSI problem is formulated in Section~\ref{sec:icsi_formulation}.
Section~\ref{sec:characterization} is devoted to the characterizations of graphs and digraphs of near-extreme min-ranks. 
We prove the hardness of the Min-Rank problem for digraphs in Section~\ref{sec:MR_hardness}.
The circuit-packing bound is established in Section~\ref{sec:circuit_packing_bound}. 
Finally, some interesting open problems are proposed in Section~\ref{sec:conclusion}. 

\section{Notation and Definitions}
\label{sec:not_def}

Let $[n]$ denote the set of integers $\{1,2,\ldots,n\}$.
Let $\fq$ denote the finite field of $q$ elements and $\fq^* = \fq \setminus \{\bO\}$. 
The \emph{support} of a vector $\bu \in \fq^n$ is defined to be the set $\text{supp}(\bu) = \{i \in [n]: u_i \ne 0\}$. 
For an $n \times k$ matrix $\bM$, let $\bM_i$ denote the $i$th row of $\bM$. 
For a set $E \subseteq [n]$, let $\bM_E$ denote the $|E| \times k$ sub-matrix of 
$\bM$ formed by rows of $\bM$ which are indexed by the elements of $E$.
For any matrix $\bM$ over $\fq$, 
we denote by $\rank(\bM)$ the rank of $\bM$ over $\fq$ (or the \emph{$q$-rank} of $\bM$). 
We use $\be_i$
to denote the unit vector, which has a one at the $i$th position, and zeros elsewhere.

A simple \emph{graph} is a pair $\G = (\VG, \EG)$ where $\VG$ is the set of vertices of $\G$
and $\EG$ is a set of \emph{unordered} pairs of distinct vertices of $\G$. We refer to $\EG$ as the set of \emph{edges} of $\G$. A typical edge of $\G$ is of the form $\{u,v\}$ where $u\in \VG$, $v \in \VG$, and $u \neq v$.
If $e = \{u,v\} \in \EG$ we say that $u$ and $v$ are adjacent. 
We also refer to $u$ and $v$ as the \emph{endpoints} of $e$.

A simple \emph{digraph}\index{digraph} is a pair $\D = (\VD, \ED)$ where $\VD$ is the set of vertices of $\D$, 
and $\ED$ is a set of \emph{ordered} pairs of distinct vertices of $\D$. 
We refer to $\ED$ as the set of arcs (or directed edges) of $\D$.
A typical \emph{arc} of $\D$ is of the form 
$e = (u,v)$ where $u \in \VD$, $v \in \VD$, and $u \neq v$.
The vertices $u$ and $v$ are called the \emph{endpoints} of the arc $e$. 

Simple graphs and digraphs have no loops and no parallel edges
and arcs, respectively. 
In the scope of this paper, only simple graphs and digraphs are considered. 
Therefore, we simply refer to them as graphs and digraphs for succinctness.  

The number of vertices $|\VD|$ is called the \emph{order} of $\D$,
whereas the number of arcs $|\ED|$ is called the \emph{size} of $\D$.
The \emph{complement} of a digraph $\D$, denoted by $\Dc$, is defined as follows.
The vertex set is $\VDc = \VD$. The arc set is
\[
\EDc = \big\{(u,v): \ u, v \in \VD, \ u \neq v, \ (u,v) \notin \ED \big\}.
\] 
Analogous concepts are also defined for graphs. 

A digraph $\D$ is called \emph{symmetric} if it satisfies the property that 
$(u,v) \in \ED$ if and only if $(v,u) \in \ED$.
A symmetric digraph can be viewed as a graph, and vice versa.  
A \emph{complete graph} is a graph that contains all possible edges. A \emph{complete digraph} is a digraph that contains all possible arcs. 

A collection of subsets $V_1, V_2, \ldots, V_k$ of a set $V$ is said to 
\emph{partition} $V$ if $\cup_{i = 1}^k V_i = V$ and $V_i \cap V_j = \varnothing$
for every $i \neq j$. In that case, $[V_1, V_2, \ldots, V_k]$ is referred
to as a partition of $V$, and $V_i$'s ($i \in [k]$) are called \emph{parts}
of the partition.

A graph $\G$ is called \emph{bipartite} if $\VG$ can 
be partitioned into two subsets $U$ and $V$ such that for every edge $\{u,v\} \in \EG$, 
it holds that $u \in U$ and $v \in V$, or vice versa.

A \emph{subgraph} of a graph $\G$ is a graph whose vertex set $V$ is a subset of that of $\G$ and whose edge set is a subset of that of $\G$ restricted to the vertices in $V$. 
Let $V$ be a subset of vertices in $\VG$. The subgraph of $\G$ \emph{induced} by $V$ is a graph whose vertex set is $V$, and edge set is $\{\{u,v\}: \ u \in V,\ v \in V, \ \{u,v\} \in \EG\}$. We refer to such a graph as an \emph{induced subgraph} of $\G$.
A subgraph and induced subgraph of a digraph can be defined in a similar manner. 

A \emph{path} in a graph $\G$ is a sequence of distinct vertices $(v_1,v_2,\ldots,v_r)$, such that $\{v_s,v_{s+1}\} \in \EG$ for all $s \in [r-1]$.
A \emph{directed path} in a digraph $\D$ is a sequence of distinct vertices $(v_1,v_2,\ldots,v_r)$, such that $(v_s,v_{s+1}) \in \ED$, for all $s \in [r-1]$.

A \emph{circuit} in a digraph $\D$ is a sequence of pairwise distinct vertices 
\[
\C = (v_1, v_2, \ldots, v_r),
\] 
where $(v_s,v_{s+1}) \in \ED$ for all $s \in [r - 1]$ and $(v_r,v_1) \in \ED$ as well. 
A digraph is called \emph{acyclic} if it contains no circuits.

A graph is called \emph{connected} if there is a path from each vertex in the graph to every other vertex.
The \emph{connected components} of a graph are its maximal connected subgraphs.
Similarly, a digraph is called \emph{strongly connected} if there is a directed path from each vertex in the graph to every other vertex.
The \emph{strongly connected components} of a digraph are its maximal strongly connected subgraphs.

If $(u,v)$ is an arc in a digraph $\D$, then $v$ is called an \emph{out-neighbor} of $u$
in $\D$.
The set of out-neighbors of a vertex $u$ in a digraph $\D$ is denoted by $N^\D_O(u)$.
We simply use $N_O(u)$ whenever there is no potential confusion.
We also denote by $N^\G(u)$ the set of neighbors of $u$ in a graph $\G$, namely, the set of vertices adjacent to $u$ in $\G$.

An \emph{independent set} in a graph $\G$ is a set of vertices of $\G$ 
with no edges connecting any two of them. An independent set in $\G$ of largest
cardinality is called a \emph{maximum independent set} in $\G$. 
The cardinality of such a maximum independent set is referred to as 
the \emph{independence number} of $\G$, denoted by $\aG$. We also use $\aD$ to denote the size of a maximum acyclic induced subgraph of a digraph $\D$
for the following reason.  
For a symmetric digraph $\D$, $\aD$ is equal to
the size of a maximum independent set if $\D$ is regarded as a graph. 

A clique of a graph is a set of vertices
that induces a complete subgraph of that graph.
A \emph{clique cover} of a graph is a set of cliques that
partition its vertex set.
A \emph{minimum clique cover} of a graph
is a clique cover with the minimum number of cliques. The number of cliques in 
such a minimum clique cover of a graph is called the clique cover number of that graph.
Similar concepts are defined for digraphs. 
We denote by $\cG$ the clique cover number of a graph $\G$ and $\cD$ the clique cover number 
of a digraph $\D$.

\section{The Index Coding with Side Information Problem}
\label{sec:icsi_formulation}

The ICSI problem is formulated as follows. 
Suppose a sender $S$ wants to send a vector $\bx = (x_1,x_2,\ldots,x_n)$, where $x_i \in \fq^t$ for all $i\in [n]$, to $n$ receivers $R_1,R_2,\ldots, R_n$. Each $R_i$ possesses some prior side information, consisting of the blocks $x_j$, $j \in \X_i \subsetneq [n]$, and is interested in receiving a single block $x_i$. The sender $S$ broadcasts a codeword $\fkE(\bx) \in \fq^\kp$, 
where $\kp$ is some positive integer, that enables each receiver $R_i$ to recover $x_i$ based on its side information. Such a mapping $\fkE: \fq^{nt} \ra \fq^{\kappa}$ is called an \emph{index code}. We refer to $t$ as the \emph{block length} and $\kp$ as the \emph{length} of the index code. The ratio $\kp/t$ is called the \emph{transmission rate} of the index code. The objective of $S$ is to find an \emph{optimal} index code, that is, an index code which has the minimum transmission rate. The index code is called \emph{linear} if $\fkE$ is a linear mapping, and \emph{nonlinear} otherwise. The index code is called \emph{scalar} if $t = 1$ and \emph{vector} if $t > 1$. The length and the transmission rate of a scalar index code ($t=1$) are identical.      

\vskip 10pt 
\begin{example}
\label{ex:icsi_instance}
Consider the following ICSI instance (Fig.~\ref{fig:ex_icsi}). 
There are five receivers ($n = 5$). 
We only consider scalar index codes in this example. 
Suppose that $x_i \in \ft$, $i \in [5]$, are five messages available
from $S$. 
For each $i \in [5]$, the receiver $R_i$ requests $x_i$ and owns certain messages 
as a priori. We have here $\X_1 = \{2\}$, $\X_2 = \{3\}$, $\X_3 = \{1,4\}$, 
$\X_4 = \{5\}$, and $\X_5 = \{2,4\}$. 

\begin{figure}[H]
\centering
\scalebox{1} 
{
\begin{pspicture}(0,-2.764237)(8.529688,1.4335759)
\pscircle[linewidth=0.04,dimen=outer](4.0581245,1.0435759){0.39}
\usefont{T1}{ptm}{m}{n}
\rput(3.992344,1.0235758){$S$}
\pscircle[linewidth=0.04,dimen=outer](1.098125,-0.7364242){0.39}
\usefont{T1}{ptm}{m}{n}
\rput(1.0323437,-0.7364242){$R_1$}
\pscircle[linewidth=0.04,dimen=outer](2.4581249,-1.4164243){0.39}
\usefont{T1}{ptm}{m}{n}
\rput(2.4123437,-1.3964243){$R_2$}
\pscircle[linewidth=0.04,dimen=outer](3.998125,-1.5564243){0.39}
\usefont{T1}{ptm}{m}{n}
\rput(3.9523437,-1.5164243){$R_3$}
\pscircle[linewidth=0.04,dimen=outer](5.5781255,-1.4164243){0.39}
\usefont{T1}{ptm}{m}{n}
\rput(5.5323434,-1.3764243){$R_4$}
\pscircle[linewidth=0.04,dimen=outer](7.058125,-0.7764243){0.39}
\usefont{T1}{ptm}{m}{n}
\rput(7.0123434,-0.7564242){$R_5$}
\usefont{T1}{ptm}{m}{n}
\rput(1.0676563,-1.2964243){requests $x_1$}
\usefont{T1}{ptm}{m}{n}
\rput(0.86265624,-1.6564243){owns $x_2$}
\usefont{T1}{ptm}{m}{n}
\rput(2.3476562,-2.0764244){requests $x_2$}
\usefont{T1}{ptm}{m}{n}
\rput(2.142656,-2.4364243){owns $x_3$}
\usefont{T1}{ptm}{m}{n}
\rput(4.107656,-2.1764243){requests $x_3$}
\usefont{T1}{ptm}{m}{n}
\rput(4.1926565,-2.5364244){owns $x_1,x_4$}
\usefont{T1}{ptm}{m}{n}
\rput(5.787656,-1.9764242){requests $x_4$}
\usefont{T1}{ptm}{m}{n}
\rput(5.5826564,-2.3364244){owns $x_5$}
\usefont{T1}{ptm}{m}{n}
\rput(7.2276564,-1.3764243){requests $x_5$}
\usefont{T1}{ptm}{m}{n}
\rput(7.3126564,-1.7364242){owns $x_2,x_4$}
\rput{-179.40352}(8.0204525,1.6232808){\psarc[linewidth=0.04](4.014451,0.7907659){0.7917863}{41.335026}{140.38962}}
\rput{-179.40352}(7.917033,2.201883){\psarc[linewidth=0.04](3.9642472,1.0803362){1.6270506}{42.013943}{137.06087}}
\rput{-179.40352}(8.104917,1.3452986){\psarc[linewidth=0.04](4.0559597,0.651555){0.2680074}{26.714457}{147.39537}}
\usefont{T1}{ptm}{m}{n}
\rput(5.5523434,0.9635757){$x_1+x_2$}
\usefont{T1}{ptm}{m}{n}
\rput(5.5523434,0.60357577){$x_2+x_3$}
\usefont{T1}{ptm}{m}{n}
\rput(5.5523434,0.22357577){$x_4+x_5$}
\end{pspicture} 
}

\caption{Example of an ICSI instance}
\label{fig:ex_icsi}
\end{figure}
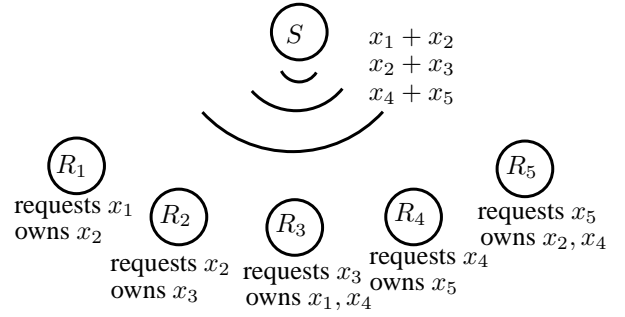

On the one hand, $S$ can satisfy the demands from all receivers 
simultaneously in a straightforward way by broadcasting all
five messages $x_i$'s, $i \in [5]$. This na\"{i}ve solution 
costs \emph{five} transmissions. On the other hand, a smarter
solution for $S$ is to broadcast \emph{three} packets $x_1+x_2$, 
$x_2 + x_3$, and $x_4+x_5$. This index code is of length three. 
The decoding process goes as follows. 
Since $R_1$ already knows $x_2$, it obtains $x_1$ by adding
$x_2$ to the first packet $x_1 + x_2$:
\[
x_1 = x_2 + (x_1+ x_2).
\]
Similarly, $R_2$ obtains $x_2 = x_3 + (x_2+ x_3)$; $R_3$ obtains 
$x_3 = x_1 + (x_1 + x_2) + (x_2 + x_3)$; $R_4$ obtains 
$x_4 = x_5 + (x_4 + x_5)$; $R_5$ obtains $x_5 = x_4 + (x_4 + x_5)$.
\end{example}
\vskip 10pt 

Each instance of the ICSI problem can be described by the so-called \emph{side information digraph} \cite{Yossef}. Given $n$ and $\X_i$, $i \in [n]$, the \emph{side information digraph} $\D = (\VG, \ED)$ is defined as follows. The vertex set $\VD = \{u_1,u_2,\ldots,u_n\}$. The edge set $\ED = \cup_{i \in [n]} \big\{(u_i,u_j):\ j \in \X_i \big\}$. 
Sometimes we simply take $\VD = [n]$ and $\ED = \cup_{i \in [n]} \big\{(i,j):\ j \in \X_i \big\}$.
If $\D$ is a symmetric digraph, we can regard $\D$ as a graph, 
and refer to $\D$ as the \emph{side information graph}. 

The side information digraph that describes the ICSI instance in 
Example~\ref{ex:icsi_instance} is depicted in Fig.~\ref{fig:sid}.
Here we choose $\VD = [5]$. 

\begin{figure}[h]
\centering
\scalebox{1} 
{
\begin{pspicture}(0,-1.82)(3.92,1.82)
\pscircle[linewidth=0.04,dimen=outer](1.97,1.47){0.35}
\usefont{T1}{ptm}{m}{n}
\rput(1.9514062,1.47){$1$}
\pscircle[linewidth=0.04,dimen=outer](3.57,0.33){0.35}
\usefont{T1}{ptm}{m}{n}
\rput(3.5514061,0.33){$2$}
\pscircle[linewidth=0.04,dimen=outer](2.79,-1.47){0.35}
\usefont{T1}{ptm}{m}{n}
\rput(2.7714062,-1.47){$3$}
\pscircle[linewidth=0.04,dimen=outer](1.09,-1.43){0.35}
\usefont{T1}{ptm}{m}{n}
\rput(1.0714062,-1.43){$4$}
\pscircle[linewidth=0.04,dimen=outer](0.35,0.33){0.35}
\usefont{T1}{ptm}{m}{n}
\rput(0.33140624,0.33){$5$}
\psline[linewidth=0.04cm,arrowsize=0.05291667cm 2.0,arrowlength=1.4,arrowinset=0.4]{->}(2.28,1.3)(3.34,0.52)
\psline[linewidth=0.04cm,arrowsize=0.05291667cm 2.0,arrowlength=1.4,arrowinset=0.4]{->}(3.46,0.0)(2.96,-1.18)
\psline[linewidth=0.04cm,arrowsize=0.05291667cm 2.0,arrowlength=1.4,arrowinset=0.4]{->}(2.66,-1.14)(2.02,1.14)
\psline[linewidth=0.04cm,arrowsize=0.05291667cm 2.0,arrowlength=1.4,arrowinset=0.4]{->}(0.68,0.34)(3.26,0.32)
\psline[linewidth=0.04cm,arrowsize=0.05291667cm 2.0,arrowlength=1.4,arrowinset=0.4]{->}(0.36,0.0)(0.86,-1.24)
\psline[linewidth=0.04cm,arrowsize=0.05291667cm 2.0,arrowlength=1.4,arrowinset=0.4]{->}(1.02,-1.08)(0.56,0.08)
\psline[linewidth=0.04cm,arrowsize=0.05291667cm 2.0,arrowlength=1.4,arrowinset=0.4]{->}(2.44,-1.46)(1.44,-1.46)
\end{pspicture} 
}
\caption{The corresponding side information digraph (Fig.~\ref{fig:ex_icsi})}
\label{fig:sid}
\end{figure}
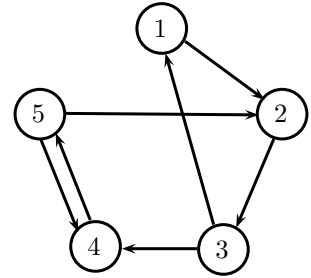 

\vskip 10pt 
\begin{definition}[\cite{Haemers1978}]
\label{def:mr-def}
Let $\D = \big(\VD,\ED \big)$ be a digraph of order $n$, where $\VD
=\{u_1,u_2,\ldots,u_n\}$.  
\begin{enumerate}
	\item A matrix $\bM=(m_{u_i,u_j}) \in \fq^{n \times n}$ (whose rows
	and columns are labeled by the elements of $\VD$) is said to \emph{fit} $\D$ if 
\[
\begin{cases} m_{u_i,u_j} \neq 0,& i = j, \\ m_{u_i,u_j} = 0,& i \neq j, \ (u_i,u_j) \notin \ED. \end{cases}
\]
  \item The \emph{min-rank} of $\D$ over $\fq$ is defined to be
\end{enumerate}	
	\[
	\mrq(\D) \define \min\left\{\rank(\bM): \ \bM \in \fq^{n \times n} \text{ and } \bM \text{ fits } \D \right\}.
	\]
Since a graph can be viewed as a symmetric digraph, 
the above definitions also apply to a graph. 
\end{definition} 
\vskip 10pt 

For instance, the two matrices in Fig.~\ref{fig:matrix_fit_graph} fit 
the digraph $\D$ depicted in Fig.~\ref{fig:sid}. The matrix $\bM_2$ 
has $2$-rank three. By Definition~\ref{def:mr-def}, $\mrt(\D) \leq 3$. 
By Theorem~\ref{thm:sandwiched_theorem_graph} stated below, as
$\aD = 3$, we deduce that $\mrt(\D) \geq 3$. Thus, $\mrt(\D) = 3$ and
$\bM_2$ achieves the min-rank. 

\begin{figure}[h]
\subfloat[A matrix of $2$-rank four]{
$
\bM_1  =
\begin{pmatrix}
1 & 1 & 0 & 0 & 0\\
0 & 1 & 1 & 0 & 0\\
0 & 0 & 1 & 1 & 0\\
0 & 0 & 0 & 1 & 1\\
0 & 1 & 0 & 0 & 1\\
\end{pmatrix}
$
}
\subfloat[A matrix of $2$-rank three]{
$
\bM_2  =
\begin{pmatrix}
1 & 1 & 0 & 0 & 0\\
0 & 1 & 1 & 0 & 0\\
1 & 0 & 1 & 0 & 0\\
0 & 0 & 0 & 1 & 1\\
0 & 0 & 0 & 1 & 1\\
\end{pmatrix}
$
}
\caption{Two matrices that fit $\D$ (Fig.~\ref{fig:sid})}
\label{fig:matrix_fit_graph}
\end{figure} 

Observe that the index code presented in Example~\ref{ex:icsi_instance}
is obtained by taking the dot products of $\bx$ with the first, the second, and 
the forth rows of $\bM_2$. 
These three rows actually span the row space of $\bM_2$. 
This index code has length three, which equals $\rankt(\bM_2)$. 
According to Theorem~\ref{thm:mr_theorem}, this index code
is an optimal scalar linear index code over $\ft$ for the ICSI instance 
described in Example~\ref{ex:icsi_instance}.

\vskip 10pt 
\begin{theorem}[\cite{Yossef, LubetzkyStav}]
\label{thm:mr_theorem}
The length of an optimal scalar linear index code over $\fq$ for the ICSI instance 
described by $\D$ is $\mrq(\D)$. 
\end{theorem}  
\vskip 10pt 

Let $\bt_q(t,\D)$ denote the length of an optimal \emph{vector} 
index code of block length $t$
over $\fq$ for an ICSI instance described by a digraph $\D$. 
Note that we do not require the index codes to be linear.    
Alon {\et} \cite{Alon} defined the \emph{broadcast rate} $\bt_q(\D)$ 
of the corresponding ICSI instance 
to be $\lim_{t \ra \infty}\bt_q(t,\D) / t$
(see also Blasiak {\et}~\cite{BlasiakKleinbergLubetzky2011})
\footnote{In~\cite{Alon} and~\cite{BlasiakKleinbergLubetzky2011}, 
the authors only consider the case $q = 2$, and therefore they use the notations $\bt_t$ and $\bt$,
which is independent of the field size. In our notations, this will correspond to $\bt_2(t,\D)$ and $\bt_2(\D)$. 
At the moment, it is not clear whether the field size $q$ plays any significant role with respect to the value of $\bt_q$. 
For example, in this work, the analysis of min-rank for the cases $q=2$ and $q>2$ is different
(thus, the result in Section~\ref{subsec:mr_two} only applies to $q = 2$). 
Therefore, in the sequel we use the subscript $q$ to ensure the consistence of the notation throughout the work.}
.
In words, the broadcast rate is the average minimum communication cost per symbol in each block $x_i$ 
(for long blocks). 
The reciprocal of $\bt_q(\D)$ is also referred to as the \emph{capacity} (over $\fq$)
of the ICSI instance described by $\D$ (see Langberg and Sprintson~\cite{LangbergSprintson2008}).
Theorem~\ref{thm:sandwiched_theorem_graph} demonstrates 
an intuitive fact that in terms of transmission rates, vector (nonlinear) index codes are at least as good as 
scalar (nonlinear) index codes, which in turn are at least as good as scalar
linear index codes. The last inequality in this theorem is 
called the \emph{clique-covering bound} for min-ranks. 

\vskip 10pt
\begin{theorem}[\cite{Haemers1978, Yossef, Yossef-journal,Alon}]
\label{thm:sandwiched_theorem_graph} 
For any digraph $\D$ we have
\[
\aD \leq \bt_q(\D) \leq \bt_q(1,\D) \leq \mrq(\D) \leq \cc(D).
\]
The same inequalities hold for graphs. 
\end{theorem} 
 
\section{Digraphs of Near-Extreme Min-Ranks}
\label{sec:characterization}

Some of the results presented below are folklore. 
However, we include their proofs for completeness.  

\subsection{(Strongly) Connected Components and Min-Ranks}
\label{subsec:connected_components}

\begin{lemma}[Folklore]
\label{lem:connected_components}
Let $\G = (\VG, \EG)$ be a graph. 
Suppose that $\G_1, \G_2, \ldots, \G_k$ are subgraphs 
of $\G$ that satisfy the following conditions
\begin{enumerate}
\item
The sets $\V(\G_i)$, $i \in [k]$, partition $\V(\G)$;
\item 
There is no edge of the form $\{u,v\}$ where $u \in \V(\G_i)$ and
$v \in \V(\G_j)$ for $i \neq j$. 
\end{enumerate}
Then
\[
\mrq(\G) = \sum_{i = 1}^k \mrq(\G_i).
\]
In particular, the above equality holds if $\G_1, \G_2, \ldots, \G_k$ are all connected components of $\G$.
\end{lemma}
\begin{proof}
The proof follows directly from the fact that a matrix fits $\G$ 
if and only if it is a block diagonal matrix (relabeling the vertices if necessary) 
and the block sub-matrices fit the corresponding subgraphs $\G_i$'s, $i \in [k]$.  
\end{proof} 
\vskip 10pt 

\begin{lemma}[Folklore]
\label{lem:strongly_connected_component}
Let $\D = (\VD, \ED)$ be a digraph. If $\D_1, \D_2, \ldots, \D_k$ are all strongly connected 
components of $\D$, then 
\[
\mrq(\D) = \sum_{i = 1}^k \mrq(\D_i).
\]
\end{lemma} 
\begin{proof} 
Suppose that $\V_i$ is the set of vertices that induces $\D_i$, $i \in [k]$. Then $\{\V_i\}_{i \in [k]}$
forms a partition of $\VD$. 
By relabeling the vertices of $\D$ if necessary, we may assume without loss of generality that for every
$i < j$
\begin{enumerate}
	\item $u < v$ whenever $u \in \V_i$ and $v \in \V_j$;	   
	\item There are no arcs of the form $(v,u)$ where $u \in \V_i$ and $v \in \V_j$. 
\end{enumerate}
If $\bM^{(i)}$ is a minimum-rank matrix that fits $\D_i$ ($i \in [k]$) then the diagonal block matrix $\bM$ whose
diagonal blocks are $\bM^{(i)}$ clearly fits $\D$. Moreover, 
\[
\rank(\bM) = \sum_{i = 1}^k \rank(\bM^{(i)}) = \sum_{i=1}^k \mrq(\D_i). 
\]
Hence $\mrq(\D) \leq \sum_{i = 1}^k \mrq(\D_i)$.
\begin{figure}[htb]
\centering{
\scalebox{1} 
{
\begin{pspicture}(0,-3.9292188)(9.321875,3.9692187)
\psframe[linewidth=0.04,dimen=outer](8.839531,3.5107813)(1.3995312,-3.9292188)
\psline[linewidth=0.04cm](1.3995312,2.4907813)(8.8,2.5092187)
\psline[linewidth=0.04cm](1.4195312,1.4907813)(8.8195305,1.4707812)
\psline[linewidth=0.04cm](1.3995312,0.49078125)(8.799531,0.47078124)
\psline[linewidth=0.04cm](1.3995312,-2.9292188)(8.8195305,-2.9692187)
\psline[linewidth=0.04cm](1.3995312,-1.9492188)(8.799531,-1.9492188)
\psline[linewidth=0.04cm](2.3595314,3.4707813)(2.3595314,-3.9092188)
\psline[linewidth=0.04cm](3.3795314,3.4907813)(3.3395312,-3.9092188)
\psline[linewidth=0.04cm](4.3595314,3.4907813)(4.3795314,-3.8892188)
\psline[linewidth=0.04cm](7.8395314,3.4707813)(7.8795314,-3.9092188)
\psline[linewidth=0.04cm](6.7595315,3.4907813)(6.7595315,-3.9092188)
\usefont{T1}{ptm}{m}{n}
\rput(1.7923437,1.9607812){$\bO$}
\usefont{T1}{ptm}{m}{n}
\rput(1.8123437,0.9607813){$\bO$}
\usefont{T1}{ptm}{m}{n}
\rput(2.8323438,0.9607813){$\bO$}
\usefont{T1}{ptm}{m}{n}
\rput(1.8123437,-2.4392188){$\bO$}
\usefont{T1}{ptm}{m}{n}
\rput(2.8123438,-2.4392188){$\bO$}
\usefont{T1}{ptm}{m}{n}
\rput(3.8323438,-2.4592187){$\bO$}
\usefont{T1}{ptm}{m}{n}
\rput(1.8123437,-3.4392188){$\bO$}
\usefont{T1}{ptm}{m}{n}
\rput(2.8123438,-3.4392188){$\bO$}
\usefont{T1}{ptm}{m}{n}
\rput(3.8123438,-3.4392188){$\bO$}
\usefont{T1}{ptm}{m}{n}
\rput(7.2323437,-3.4392188){$\bO$}
\psdots[dotsize=0.12](1.8395312,-0.32921875)
\psdots[dotsize=0.12](1.8395312,-0.70921874)
\psdots[dotsize=0.12](1.8395312,-1.1092187)
\psdots[dotsize=0.12](2.8595314,-0.32921875)
\psdots[dotsize=0.12](2.8595314,-0.70921874)
\psdots[dotsize=0.12](2.8595314,-1.1092187)
\psdots[dotsize=0.12](3.8595314,-0.32921875)
\psdots[dotsize=0.12](3.8595314,-0.70921874)
\psdots[dotsize=0.12](3.8595314,-1.1092187)
\psdots[dotsize=0.12](5.159531,-0.30921876)
\psdots[dotsize=0.12](5.579531,-0.7292187)
\psdots[dotsize=0.12](5.9395313,-1.1092187)
\psdots[dotsize=0.12](7.2995315,-0.32921875)
\psdots[dotsize=0.12](7.2995315,-0.70921874)
\psdots[dotsize=0.12](7.2995315,-1.1092187)
\psdots[dotsize=0.12](8.299531,-0.30921876)
\psdots[dotsize=0.12](8.299531,-0.68921876)
\psdots[dotsize=0.12](8.299531,-1.0892187)
\psdots[dotsize=0.12](5.139531,2.9907813)
\psdots[dotsize=0.12](5.559531,2.9907813)
\psdots[dotsize=0.12](5.9595313,2.9907813)
\psdots[dotsize=0.12](5.139531,1.9707812)
\psdots[dotsize=0.12](5.559531,1.9707812)
\psdots[dotsize=0.12](5.9595313,1.9707812)
\psdots[dotsize=0.12](5.179531,-2.4292188)
\psdots[dotsize=0.12](5.599531,-2.4292188)
\psdots[dotsize=0.12](5.9995313,-2.4292188)
\psdots[dotsize=0.12](5.179531,-3.4092188)
\psdots[dotsize=0.12](5.599531,-3.4092188)
\psdots[dotsize=0.12](5.9995313,-3.4092188)
\usefont{T1}{ptm}{m}{n}
\rput(1.8923438,3.7807813){$\V_1$}
\usefont{T1}{ptm}{m}{n}
\rput(2.8323438,3.7807813){$\V_2$}
\usefont{T1}{ptm}{m}{n}
\rput(3.8523438,3.7807813){$\V_3$}
\usefont{T1}{ptm}{m}{n}
\rput(7.262344,3.7807813){$\V_{k-1}$}
\usefont{T1}{ptm}{m}{n}
\rput(8.252344,3.7807813){$\V_k$}
\usefont{T1}{ptm}{m}{n}
\rput(0.6923438,3.0007813){$\V_1$}
\usefont{T1}{ptm}{m}{n}
\rput(0.71234375,1.9807812){$\V_2$}
\usefont{T1}{ptm}{m}{n}
\rput(0.73234373,0.98078126){$\V_3$}
\usefont{T1}{ptm}{m}{n}
\rput(0.82234377,-2.4192188){$\V_{k-1}$}
\usefont{T1}{ptm}{m}{n}
\rput(0.6923438,-3.4192188){$\V_k$}
\usefont{T1}{ptm}{m}{n}
\rput(1.8714062,3.0192187){$\bM^{(1)}$}
\usefont{T1}{ptm}{m}{n}
\rput(2.8914063,1.9592187){$\bM^{(2)}$}
\usefont{T1}{ptm}{m}{n}
\rput(3.8514063,0.9792188){$\bM^{(3)}$}
\usefont{T1}{ptm}{m}{n}
\rput(7.3214064,-2.440781){\small $\bM^{(k-1)}$}
\usefont{T1}{ptm}{m}{n}
\rput(8.331407,-3.440781){$\bM^{(k)}$}
\end{pspicture} 
}
}
\caption{Matrix $\bM$ that fits $\D$}
\label{fig:matrix_fittingG}
\end{figure}
It remains to show that $\mrq(\D) \geq \sum_{i = 1}^k \mrq(\D_i)$. 
Suppose that the matrix $\bM$ fits $\D$. 
By the assumptions on $\V_i$'s ($i \in [k]$) stated at the beginning of the proof,
$\bM$ must be an upper-triangular block matrix, as 
shown in Fig.~\ref{fig:matrix_fittingG}. If we let $\bM^{(i)}$ be the sub-matrix of $\bM$
formed by the rows and columns indexed by the elements of $\V_i$, then $\bM^{(i)}$ fits $\D_i$ and hence,
\[
\rank(\bM) \geq \sum_{i = 1}^k \rank(\bM^{(i)}) \geq \sum_{i=1}^k \mrq(\D_i).
\]
Thus, $\mrq(\D) \geq \sum_{i=1}^k \mrq(\D_i)$. 
\end{proof} 
\vskip 10pt 

These two lemmas suggest that it is sufficient to study the min-ranks of 
connected graphs and strongly connected digraphs, respectively.

\subsection{Digraphs of Min-Rank One}

\begin{proposition}[Folklore]
\label{lem:complete_graph_pp3}
Let $\D = (\VD, \ED)$ be a digraph. Then $\mrq(\D) = 1$ if and only if 
$\D$ is a complete digraph. The same statement holds for a graph.  
\end{proposition} 
\begin{proof}
Suppose $\D$ is a digraph of order $n$. 
If $\mrq(\D) = 1$, by the definition of min-rank there exists
an $n \times n$ matrix $\bM = (m_{u,v})$ of $q$-rank one that fits $\D$. 
Then the rows of $\bM$ must be scalar multiples of each other. 
Moreover, $m_{u,u} \neq 0$ for all $u \in \VD$. 
Hence $m_{u,v} \neq 0$ for all $u \in \VD$ and all $v \in \VD$. 
Therefore, $(u,v) \in \ED$ for all $u \neq v$, $u \in \VD$ and $v \in \VD$. 
In other words, $\D$ is a complete digraph. 

Conversely, suppose that $\D$ is a complete digraph. Then $\bJ$, the $n \times n$ 
all-one matrix, fits $\D$ and $\mrq(\bJ) = 1$, which implies that $\mrq(\D) = 1$. 
The same arguments hold for graphs.  
\end{proof} 
\vskip 10pt 

\begin{corollary} 
Let $\D = (\VD, \ED)$ be a digraph. Then $\bt_q(\D) = 1$ if and only if 
$\D$ is a complete digraph. The same statement holds for a graph.  
\end{corollary} 
\begin{proof} 
Suppose $\bt_q(\D) = 1$. Then by Theorem~\ref{thm:sandwiched_theorem_graph}, 
$\aD = 1$. Therefore, $\D$ is a complete digraph. 
Conversely, if $\D$ is a complete digraph then by Proposition~\ref{lem:complete_graph_pp3}, 
$\mrq(\D) = 1$. Again by Theorem~\ref{thm:sandwiched_theorem_graph}, 
$\bt_q(\D) = 1$. 
\end{proof} 

\subsection{Digraphs of Min-Rank Two}
\label{subsec:mr_two}

In this section, only the \emph{binary} alphabet is considered.
We first introduce the following concept
of a \emph{fair coloring} of a digraph. Recall that a $k$-coloring 
of a graph $\G = (\VG, \EG)$ is a mapping $\phi: \VG \ra [k]$ which satisfies 
the condition that $\phi(u) \neq \phi(v)$ whenever $\{u,v\} \in \EG$. 
We often refer to $\phi(u)$ as the \emph{color} of $u$. 
If there exists a $k$-coloring of $\G$, then we say that $\G$ is $k$-colorable.

\vskip 5pt 
\begin{definition}
\label{def:fair_coloring}
Let $\D = (\VD, \ED)$ be a digraph. A \emph{fair $k$-coloring} of $\D$ is a mapping $\phi: \VD
\ra [k]$ that satisfies the following conditions:
\begin{enumerate}
	\item[(C1)] If $(u,v) \in \ED$ then $\phi(u) \neq \phi(v)$;
	\item[(C2)] For each vertex $u$ of $\D$, it holds that $\phi(v) = \phi(\omega)$ for all out-neighbors $v$ and $\omega$ of $u$.  
\end{enumerate}  
If there exists a fair $k$-coloring of $\D$, we say that we can \emph{color} $\D$ \emph{fairly by $k$ colors}, 
or, $\D$ is \emph{fairly $k$-colorable}.  
\end{definition} 
\vskip 3pt 

We refer to the condition (C2) as the \emph{fairness} of the coloring, 
since this condition guarantees that all out-neighbors of each vertex share the same color. 

\vskip 5pt 
\begin{lemma}
\label{lem:fair_coloring}
A digraph $\D = (\VD, \ED)$ is fairly $3$-colorable if and only if there exists
a partition of $\VD$ into three subsets $A$, $B$, and $C$ that satisfy the following 
conditions
\begin{enumerate}
	\item For every $u \in A$: either $N_O(u) \subseteq B$ or $N_O(u) \subseteq C$;
	\item For every $u \in B$: either $N_O(u) \subseteq A$ or $N_O(u) \subseteq C$;
	\item For every $u \in C$: either $N_O(u) \subseteq A$ or $N_O(u) \subseteq B$.
\end{enumerate} 
\end{lemma}
\begin{proof}
If $\D$ is fairly $3$-colorable, let $A$, $B$, and $C$ respectively be the sets of vertices of $\D$
that share the same color. 
Then clearly $A$, $B$, and $C$ partition $\VD$. 
Moreover, since all out-neighbors of each vertex must have the same color, 
the three conditions above are obviously satisfied. Conversely, if those conditions are satisfied, then 
$\phi: \ \VD \ra [3]$, defined by 
\[
\phi(u) = 
\begin{cases} 
1, & u \in A\\
2, & u \in B\\
3, & u \in C
\end{cases}, 
\]
is a fair $3$-coloring of $\D$.  
\end{proof} 
\vskip 10pt 

\begin{theorem}
\label{thm:mr_color}
Let $\D = (\VD, \ED)$ be a digraph. Then $\mrt(\D) \leq 2$ if and only if $\Dc$, 
the complement of $\D$, is fairly $3$-colorable. 
\end{theorem}
\begin{proof} 
\mbox{}\\
\nin {\bf The ONLY IF direction:}\\
By the definition of min-rank, $\mrt(\D) \leq 2$ implies the existence of an $n \times n$ binary matrix $\bM$
of $2$-rank at most two that fits $\D$. There must be some two rows of $\bM$ that 
span its entire row space. Without loss of generality, suppose that they are the first two rows of $\bM$, 
namely, $\bM_1$ and $\bM_2$ (these two rows might be 
linearly dependent if $\mrt(\D) < 2$). Let $A$, $B$, and $C$ be disjoint subsets of $\VD$ such that
\[
\supp(\bM_1) = A \cup B,\ \supp(\bM_2) = B \cup C.
\]
Hence, 
\[
\supp(\bM_1) \cap \supp(\bM_2) = B. 
\]
Since the binary alphabet is considered and the matrix $\bM$ has no zero rows, 
for every $u \in \VD$, one of the following must hold: 
(1) $\bM_u = \bM_1$; (2) $\bM_u = \bM_2$; 
(3) $\bM_u = \bM_1 + \bM_2$. 
Hence for every $u \in \VD$
\[
u \in \supp(\bM_u) \subseteq A \cup B \cup C. 
\]
This implies that $A \cup B \cup C = \VD$. 

Suppose that $u \in A$. Then either $\bM_u = \bM_1$ or $\bM_u = \bM_1 + \bM_2$. 
The former condition holds 
if and only if $\supp(\bM_u) = A \cup B$, which in turns implies that $(u,v) \in \ED$ for all 
$v \in A\cup B \setminus \{u\}$. 
In other words, $(u,v) \notin \EDc$ for all $v \in A\cup B$. 
Here $\Dc = (\V(\Dc), \E(\Dc))$ is the complement of $\D$. 
The latter condition holds if and only if  
$\supp(\bM_u) = A \cup C$, which implies that 
$(u,v) \notin \EDc$ for all $v \in A\cup C$. In summary, for every $u \in A$ we have
\begin{enumerate}
	\item $(u,v) \notin \EDc$, for all $v \in A$;
  \item Either $(u,v) \notin \EDc$ for all $v \in B$, or $(u,v) \notin \EDc$ for all $v \in C$; 
\end{enumerate}
In other words, for every $u \in A$, either $\N^{\Dc}_O(u) \subseteq B$ 
or $\N^{\Dc}_O(u) \subseteq C$. 
Analogous conditions hold for every $u \in B$ and for every $u \in C$ as well. Therefore, 
by Lemma~\ref{lem:fair_coloring}, $\Dc$ is fairly $3$-colorable.  \\
\nin {\bf The IF direction:}\\
Suppose now that $\Dc$ is fairly $3$-colorable. 
It suffices to find an $n \times n$ binary matrix $\bM$
of rank at most two that fits $\D$. 
By Lemma~\ref{lem:fair_coloring}, there exists a partition of $\VDc$ 
into three subsets $A$, $B$, and $C$ that satisfy the following three conditions
\begin{enumerate}
	\item For every $u \in A$: either $N^{\Dc}_O(u) \subseteq B$ or $N^{\Dc}_O(u) \subseteq C$;
	\item For every $u \in B$: either $N^{\Dc}_O(u) \subseteq A$ or $N^{\Dc}_O(u) \subseteq C$;
	\item For every $u \in C$: either $N^{\Dc}_O(u) \subseteq A$ or $N^{\Dc}_O(u) \subseteq B$.
\end{enumerate} 
We construct an $n \times n$ matrix $\bM = (m_{u,v})$ as follows. 
For each $u \in A$, if $N^{\Dc}_O(u) \subseteq B$ then let
\[
m_{u,v} =  
\begin{cases}
1, & v \in A \cup C\\
0, & v \in B
\end{cases}.
\]
Otherwise, if $N^{\Dc}_O(u) \subseteq C$ then let
\[
m_{u,v} = 
\begin{cases}
1, & v \in A \cup B\\
0, & v \in C
\end{cases}.
\]
For $u \in B$ and $u \in C$, $\bM_u$ can be constructed analogously. 
It is obvious that $\bM$ fits $\D$. Moreover, each row of $\bM$ can
always be written as a linear combination of the two binary vectors whose 
supports are $A \cup B$ and $B \cup C$, respectively. Therefore, $\rankt(\bM)
\leq 2$. The proof is complete.  
\end{proof}
\vskip 10pt 

The following corollary characterizes the digraphs of min-rank two over $\ft$. 

\vskip 10pt 
\begin{corollary}
\label{coro:mr_color}
A digraph $\D$ has min-rank two over $\ft$ if and only if $\Dc$ is fairly $3$-colorable and $\D$ is not
a complete digraph.   
\end{corollary} 
\vskip 10pt 

For a graph $\G$, it was proved by Blasiak {\et}~\cite{BlasiakKleinbergLubetzky2011}
that $\bt_2(\G) = 2$ if and only if $\Gc$ is bipartite and $\G$ is not a complete graph.
A characterization of digraphs $\D$ with $\bt_2(\D) = 2$ was also obtained therein.  
More specifically, it was shown that $\bt_2(\D) = 2$ if and only if 
$\Dc$ does not contain a subgraph isomorphic to an \emph{almost alternating cycle}. 
The almost alternating $(2m+1)$-cycle ($m \geq 1$) is defined as follows. Its vertex set
consists of all integers between $-m$ and $m$, inclusive, and there is an edge from $i$
to $j$ if and only if $j - i \in \{m,m+1\}$. Based on this characterization, a polynomial
time algorithm to recognize a digraph $\D$ with $\bt_2(\D) = 2$ was also derived in~\cite{BlasiakKleinbergLubetzky2011}. 
Hence, the question whether an optimal \emph{vector nonlinear} index code 
of length \emph{two} 
exists for an ICSI instance described by a digraph can be answered in polynomial time. 
For \emph{scalar linear} index code, the same question turns out to be hard. 
We prove later in Section~\ref{sec:MR_hardness} that the decision problem 
whether $\mrt(\D) = 2$ is NP-complete.  

\subsection{Digraphs of Min-Ranks Equal to Their Orders}
\label{subsec:mr=n}

To tackle \emph{graphs} of min-ranks almost equal to their orders (Section~\ref{subsec:mr=n}, 
\ref{subsec:mr=n-1}, \ref{subsec:mr=n-2}), 
we employ the concept of \emph{maximum matching} from graph theory.   

\begin{definition}
\label{def:matching}
A \emph{matching} in a graph is a set of edges without common vertices. 
A \emph{maximum matching} is a matching that contains the largest
possible number of edges. The number of edges in a maximum matching in $\G$ is denoted by $\mG$. 
\end{definition}
\vskip 10pt

The following upper bound
on min-rank, so-called the \emph{maximum-matching bound}, is a 
weakened version of the clique-covering bound (see Theorem~\ref{thm:sandwiched_theorem_graph}).

\vskip 10pt 
\begin{proposition}[Maximum-matching bound]
\label{prop:mmb}
For any graph $\G$ of order $n$, it holds that $\mrq(\G) \leq n - \mG$.  
\end{proposition}
\begin{proof}
As the set of vertices of $\G$ can be covered by $\mG$ cliques of size two (the edges
in a maximum matching) and $n - 2\mG$ cliques of size one (the 
remaining vertices that are not covered by the edges in the matching), 
by Theorem~\ref{thm:sandwiched_theorem_graph}, the proof follows. 
\end{proof} 
\vskip 10pt

Graphs $\G$ that satisfy $\aG = n - \mG$ are called
Koenig-Egervary graphs~\cite{Deming1979}. 
It was proved therein that there is a polynomial time algorithm to recognize
a Koenig-Egervary graph $\G$ and subsequently find $\mG$. 
By Theorem~\ref{thm:sandwiched_theorem_graph} and Proposition~\ref{prop:mmb}, 
if $\G$ is a Koenig-Egervary graph then $\mrq(\G) = n - \mG$.
Moreover, $\mrq(\G)$ can be found in polynomial time. 
The graphs that satisfy the conditions stated in 
Proposition~\ref{pro:undirected_rank_equal_order},
Proposition~\ref{thm:star_graph}, and Theorem~\ref{thm:mr_n-2}
are all Koenig-Egervary graphs (see their proofs). 

\begin{proposition}[Folklore]
\label{pro:undirected_rank_equal_order}
Let $\G$ be a graph of order $n$. Then $\mrq(\G) = n$ if and only if $\mG = 0$ 
(or equivalently, $\G$ has no edges).  
\end{proposition}
\begin{proof} 
If $\G$ has no edges, a matrix fits $\G$ if and only if it is a diagonal matrix, whose 
entries on the main diagonal are all nonzero. The $q$-rank of such a matrix is $n$. 
Therefore, $\mrq(\G) = n$. 

Suppose for contradiction that $\mrq(\G) = n$ and $\G$ contains some edge. 
Then $\mG \geq 1$ and we have $\mrq(\G) \leq n - 1$, according to the maximum-matching
bound. We obtain a contradiction.   
\end{proof}
\vskip 10pt  

\begin{proposition}
\label{thm:mr_n}
Let $\D$ be a digraph of order $n$. Then $\mrq(\D) = n$ if and only if $\D$ is acyclic.  
\end{proposition}
\begin{proof}
Equivalently, we show that $\mrq(\D) \leq n - 1$ if and only if $\D$ has a circuit. 

Suppose that $\D$ has a circuit. Then by the circuit-packing bound established in 
Section~\ref{subsec:circuit_packing_bound}, we deduce that $\mrq(\D) \leq n - 1$. 

Conversely, suppose that $\mrq(\D) \leq n - 1$. 
Then there exists a matrix $\bM$ fitting $\D$ whose rows are linearly dependent. 
In other words, $\sum_{i \in I} \al_i\bM_i = \bO$ for some nonempty subset $I \subseteq \VD$
and for some $\al_i \in \fq^*$, $i \in I$. 
Let $\D'$ be the subgraph of $\D$ induced by the vertices in $I$ and $\bM'$
the sub-matrix of $\bM$ restricted to the rows and columns indexed by the elements of $I$. 
Obviously $\bM'$ fits $\D'$. We show that there exists a circuit in $\D'$. Since
$\sum_{i \in I} \al_i\bM'_i = \bO$,
each column of $\bM'$ has at least two nonzero entries. Therefore, 
for each vertex $v$ of $\D'$, there exists another vertex $u$ of $\D'$ such that
$(u,v)$ is an arc in $\D'$. Starting from an arbitrary vertex $v_1$ of $\D'$ and 
applying this property recursively, we obtain a sequence of vertices in $\D'$
\[
v_1, v_2, \ldots, v_s, v_{s+1}, \ldots,
\]  
where $(v_{s+1}, v_s)$ is an arc in $\D'$ for every $s \geq 1$. Since $\D'$ has a finite
number of vertices, there must be a point 
when a vertex appears twice in the above sequence for the first time. 
This vertex, together with the other vertices lying between its two occurrences, 
form a circuit inside $\D'$, which is also a circuit inside $\D$. 
\end{proof} 
\vskip 10pt 
 
The existence of a circuit in a digraph can be detected by using a depth-first search, the 
time complexity of which is linear in the size of the digraph. Hence, 
as a consequence of Proposition~\ref{thm:mr_n}, the decision problem whether a digraph has 
min-rank equal to its order can be solved in polynomial time. 

\vskip 10pt 
\begin{remark}
The second direction in the proof of Proposition~\ref{thm:mr_n} has a shorter proof as follows.
Suppose that $\mrq(\D) \leq n - 1$ but $\D$ is acyclic. 
Then $\mrq(\D) \geq \aD = n$, by Theorem~\ref{thm:sandwiched_theorem_graph}.
That is a contradiction. 
However, the original proof of Proposition~\ref{thm:mr_n}
provides us with a simple and direct proof of the inequality
$\aD \leq \mrq(\D)$ (see Corollary~\ref{coro:aD}).
This inequality for digraphs was proved indirectly via the use of 
$\bt_q(1,\D)$ by Bar-Yossef {\et}~\cite{Yossef-journal}. 
In such an indirect proof, either arguments from Information Theory
are invoked~\cite[Theorem 7]{Yossef-journal} or the corresponding
confusion graph is considered~\cite[Lemma 37]{Yossef-journal}. 
\end{remark}

\vskip 10pt 
\begin{corollary}
\label{coro:aD}
For a digraph $\D$ we have
\[
\aD \leq \mrq(\D).  
\] 
\end{corollary}  
\begin{proof}
First note that if $\D'$ is an induced subgraph of $\D$ then $\mrq(\D') \leq \mrq(\D)$. 
Indeed, suppose $\bM$ is a matrix that fits $\D$ and has rank equal to 
the min-rank of $\D$. Then the sub-matrix $\bM'$ of $\bM$ restricted
to the rows and columns indexed by the vertices in $\V(\D')$
is a matrix that fits $\D'$. Then 
\[
\mrq(\D') \leq \rank(\bM') \leq \rank(\bM) = \mrq(\D). 
\]
Now let $\D'$ be a maximum acyclic induced subgraph of $\D$ of order $\aD$. 
Since $\D'$ is acyclic, by Proposition~\ref{thm:mr_n} we have
\[
\mrq(\D) \geq \mrq(\D') = |\V(\D')| = \aD. \qedhere
\]
\end{proof} 
\vskip 10pt 

\begin{corollary}
For a digraph $\D$, $\bt_q(\D) = |\VD|$ if and only if $\D$ is acyclic. 
For a graph $\G$, $\bt_q(\G) = |\VG|$ if and only if $\G$ has no edges.  
\end{corollary} 
\begin{proof}
Suppose $\bt_q(\D) = |\VD|$. By Theorem~\ref {thm:sandwiched_theorem_graph}, 
$\mrq(\D) = |\VD|$. Therefore, $\D$ is acyclic according to Proposition~\ref{thm:mr_n}. 
Conversely, if $\D$ is acyclic then $\bt_q(\D) \geq \aD = |\VD|$. 
Similar arguments hold for graphs. 
\end{proof} 

\subsection{Graphs of Min-Ranks One Less Than Their Orders}
\label{subsec:mr=n-1}

In this section, we consider (undirected) graphs.
The corresponding case for digraphs is open.  
For a connected graph $\G$ of order at least two, it is easy to see that
$\mG = 1$ if and only if it is a \emph{star graph}, which is defined as follows. 

\begin{definition}   
A graph $\G = (\VG, \EG)$ is called a \emph{star graph}
if $|\VG| \geq 2$ and there exists a vertex $v \in \VG$ such that 
$\EG = \big\{\{u,v\}: \ u \in \VG \setminus \{v\}\big\}$. 
\end{definition} 

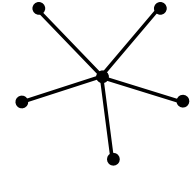
\begin{figure}[htb]
\centering{
\scalebox{0.8} 
{
\begin{pspicture}(0,-1.4158349)(3.000441,1.4422222)
\psdots[dotsize=0.21200001](0.42417023,1.2762222)
\psdots[dotsize=0.21200001](2.4441702,1.2762222)
\psdots[dotsize=0.21200001,dotangle=-27.662416](0.14309572,-0.28401718)
\psdots[dotsize=0.21200001](1.4841702,0.13622223)
\psdots[dotsize=0.21200001,dotangle=31.584677](2.8194814,-0.26963037)
\psline[linewidth=0.06cm](1.5241703,0.21622223)(2.4041703,1.2562222)
\psline[linewidth=0.06cm](1.5290428,0.11085437)(2.7474189,-0.26698118)
\psline[linewidth=0.06cm](1.4441702,0.17622222)(0.42417023,1.2362223)
\psline[linewidth=0.06cm](1.5023404,0.15513396)(0.1970941,-0.2671597)
\psdots[dotsize=0.21200001,dotangle=-33.780685](1.6653426,-1.2351804)
\psline[linewidth=0.06cm](1.4733033,0.096405506)(1.6377128,-1.1685725)
\end{pspicture} 
}
}
\caption{A star graph}
\label{fig:stargraph}
\end{figure}

It is straightforward to see that if $\mG = 1$ then $\aG = n - 1$, as $\G$ is a star graph.  

\begin{proposition}
\label{thm:star_graph}
Let $\G$ be a connected graph of order $n \geq 2$. Then $\mrq(\G) = n - 1$ if and only if 
$\mG = 1$ (or equivalently, $\G$ is a star graph).  
\end{proposition}
\begin{proof}
We first suppose that $\mrq(\G) = n - 1$. 
By the maximum-matching bound, $n - 1 = \mrq(\G) \leq n - \mG$. 
Therefore, $\mG \leq 1$. 
However, as $\mrq(\G) \neq n$, by Proposition~\ref{pro:undirected_rank_equal_order}
we have $\mG \neq 0$. Hence, $\mG = 1$. 

Conversely, assume that $\mG = 1$.
By the maximum-matching bound, $\mrq(\G) \leq n-1$. 
By Theorem~\ref{thm:sandwiched_theorem_graph}, $\mrq(\G) \geq \al(\G) = n - 1$.
Thus, $\mrq(\G) = n - 1$.   	
\end{proof} 
\vskip 10pt 

\begin{corollary}
Let $\G$ be a connected graph of order $n \geq 2$. 
Then $\bt_q(\G) = n-1$ 
if and only if $\mG = 1$ ($\G$ is a star graph). 
\end{corollary} 
\begin{proof}
Suppose $\bt_q(\G) = n-1$. Then either $\mrq(\G) = n - 1$
or $\mrq(\G) = n$. However, by Proposition~\ref{pro:undirected_rank_equal_order}, 
$\mrq(\G) = n$ implies that $\G$ has no edge. As a consequence, 
$\bt_q(\G) \geq \aG = n$, which contradicts our assumption. 
Hence, $\mrq(\G) = n - 1$. According to Proposition~\ref{thm:star_graph}, 
$\mG = 1$. 

Conversely, suppose that $\mG = 1$. According to Proposition~\ref{thm:star_graph},
we have
\[
n - 1 = \aG \leq \bt_q(\G) \leq \mrq(\G) = n-1.
\] 
Hence, $\bt_q(\G) = n - 1$. 
\end{proof} 

\subsection{Graphs of Min-Ranks Two Less Than Their Orders}
\label{subsec:mr=n-2}

In this section, we consider (undirected) graphs.
The corresponding case for digraphs is open. 
Here we also employ 
the matching language to characterize graphs of min-ranks
two less than their orders. 

\vskip 10pt 
\begin{theorem} 
\label{thm:mr_n-2}
Suppose $\G$ is a connected graph of order $n \geq 6$. 
Then $\mrq(\G) = n - 2$ if and only if $\mG = 2$ 
and $\G$ does not contain a subgraph isomorphic to 
the graph $F$ depicted in Fig.~\ref{fig:forbidden_subgraph}.  

\begin{figure}[h]
\centering
\scalebox{1} 
{
\begin{pspicture}(0,-0.5674989)(2.196694,0.5674989)
\psline[linewidth=0.04cm,dotsize=0.07055555cm 2.0]{**-**}(0.02,0.54749894)(0.02,-0.53250104)
\psline[linewidth=0.04cm,dotsize=0.07055555cm 2.0]{**-}(0.9181049,-0.050460596)(0.0,0.48749894)
\psline[linewidth=0.04cm,dotsize=0.07055555cm 2.0]{**-}(1.5980058,-0.008810784)(0.8980079,-0.010509237)
\psline[linewidth=0.04cm](0.89791083,0.029490646)(0.02,-0.45250106)
\psline[linewidth=0.04cm,dotsize=0.07055555cm 2.0]{-**}(1.5579575,0.0110921)(2.176694,0.5325949)
\psline[linewidth=0.04cm,dotsize=0.07055555cm 2.0]{-**}(1.538006,-0.008956365)(2.1393144,-0.54749894)
\end{pspicture} 
}
\caption{The forbidden subgraph $F$}
\label{fig:forbidden_subgraph}
\end{figure}
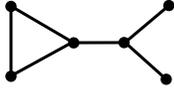
\end{theorem} 

The proof of this theorem appears in Appendix.

\begin{corollary}
If $\mG = 2$ and $\G$ contains a subgraph isomorphic to $F$ (Fig.~\ref{fig:forbidden_subgraph}) then $\mrq(\G) = |\VG| - 3$.  
\end{corollary} 
\begin{proof}
Suppose $F'$ (Fig.~\ref{fig:F'}) is a subgraph of $\G$ that is isomorphic to $F$. 

\begin{figure}[h]
\centering
\scalebox{1} 
{
\begin{pspicture}(0,-0.698125)(3.2828126,0.698125)
\psline[linewidth=0.04cm,dotsize=0.07055555cm 2.0]{**-**}(0.5809375,0.5796875)(0.5809375,-0.5003125)
\psline[linewidth=0.04cm,dotsize=0.07055555cm 2.0]{**-}(1.4790424,-0.018272037)(0.5609375,0.5196875)
\psline[linewidth=0.04cm,dotsize=0.07055555cm 2.0]{**-}(2.1589434,0.023377776)(1.4589454,0.021679323)
\psline[linewidth=0.04cm](1.4588484,0.061679207)(0.5809375,-0.4203125)
\psline[linewidth=0.04cm,dotsize=0.07055555cm 2.0]{-**}(2.118895,0.04328066)(2.7376316,0.56478345)
\psline[linewidth=0.04cm,dotsize=0.07055555cm 2.0]{-**}(2.0989435,0.023232196)(2.7002518,-0.5153104)
\usefont{T1}{ptm}{m}{n}
\rput(0.22234374,0.5096875){$a$}
\usefont{T1}{ptm}{m}{n}
\rput(0.23234375,-0.3903125){$b$}
\usefont{T1}{ptm}{m}{n}
\rput(1.3623438,-0.2703125){$c$}
\usefont{T1}{ptm}{m}{n}
\rput(2.0323439,-0.2703125){$d$}
\usefont{T1}{ptm}{m}{n}
\rput(2.9823437,0.5096875){$f$}
\usefont{T1}{ptm}{m}{n}
\rput(2.9823437,-0.4703125){$g$}
\end{pspicture} 
}
\caption{The subgraph $F'$}
\label{fig:F'}
\end{figure}
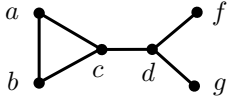

As $\G$ does not have a matching of size three, each of the vertices
$c$, $f$, and $g$ is not adjacent to any vertex in $\VG \setminus \V(F')$. 
Moreover, no pairs of vertices in $\VG \setminus \V(F')$ are adjacent
for the same reason. Therefore, $\{c,f,g\} \cup (\VG \setminus \V(F'))$
is an independent set of size $|\VG| - 3$ in $\G$. Hence, $\mrq(\G) \geq 
\aG \geq |\VG| - 3$. As $\mG = 2$, by the maximum-matching bound,
$\mrq(\G) \leq |\VG| - 2$.  
As $\G$ contains $F'$, which is isomorphic to $F$, by Theorem~\ref{thm:mr_n-2}, 
$\mrq(\G) \neq |\VG| - 2$. Thus, $\mrq(\G) = |\VG| - 3$.  
\end{proof} 
\vskip 10pt 

\begin{corollary}
Theorem~\ref{thm:mr_n-2} holds verbatim if we replace $\mrq(\cdot)$
by $\bt_q(\cdot)$. 
\end{corollary}
 
\begin{proof} 
Suppose that $\bt_q(\G) = n - 2$. Then $\mrq(\G) \in \{n-2,n-1,n\}$. 
By Proposition~\ref{pro:undirected_rank_equal_order}, Proposition~\ref{thm:star_graph}, and their corollaries, for $\kp \in \{n-1,n\}$, 
$\mrq(\G) = \kp$ if and only if $\bt_q(\G) = \kp$. 
Therefore, $\mrq(\G) = n - 2$. 
According to Theorem~\ref{thm:mr_n-2}, $\mG = 2$ and $\G$ does not contain a subgraph isomorphic to $F$. 

Conversely, as shown in the proof of Theorem~\ref{thm:mr_n-2} (the IF direction),
$\aG = \mrq(\G) = n - 2$. Therefore, $\bt_q(\G) = n -2$ by Theorem~\ref{thm:sandwiched_theorem_graph}.
\end{proof} 

\section{The Hardness of the Min-Rank Problem for Digraphs}
\label{sec:MR_hardness}

In this section, we first prove that it is an NP-complete problem to decide whether a given digraph 
is fairly $k$-colorable (see Definition~\ref{def:fair_coloring}), for any given $k \geq 3$. The hardness of this problem, 
by Lemma~\ref{lem:complete_graph_pp3} and Corollary~\ref{coro:mr_color}, 
leads to the hardness of the decision problem whether a given digraph has 
min-rank two over $\ft$. 
The fair $k$-coloring problem is defined formally as follows. 

\vskip 20pt
\begin{center} 
\fbox{
\parbox{3.2in}{
\vskip 3pt
{\it Problem}:\ {\bf FAIR $\boldsymbol{k}$-COLORING}\\
{\it Input}:\ A digraph $\D$, an integer $k$\\
{\it Output}:\ True if $\D$ is fairly $k$-colorable, False otherwise
\vskip 3pt
}}
\end{center}

\vskip 10pt 
\begin{theorem}
\label{thm:fair_coloring_NP}
The fair $k$-coloring problem is NP-complete for $k \geq 3$.  
\end{theorem}   
\begin{proof} 
This problem is obviously in NP, as the algorithm 
can guess a candidate for the fair coloring and verify that the candidate is indeed a 
fair coloring in polynomial time. 
For NP-hardness, we reduce the $k$-coloring problem to the fair $k$-coloring problem.
Recall that the $k$-coloring problem is the decision problem whether a given graph is $k$-colorable. 
Suppose that $\G = (\VG, \EG)$ is an arbitrary graph. We aim to build a digraph 
$\D = (\VD, \ED)$ so that $\G$ is $k$-colorable if and only if $\D$ is fairly $k$-colorable. 
Suppose that $\VG = [n]$. For each vertex $i \in [n]$, we build the following gadget, which is a digraph 
$\D_i = (\V_i, \E_i)$. The vertex set of $\D_i$ is
\[
\V_i = \{i\} \cup \big\{\omega_{i,j}: \ j \in N^{\G}(i)\big\},  
\]
where $\omega_{i,j}$ are newly introduced vertices. We refer to $\omega_{i,j}$
as a \emph{clone} (in $\D_i$) of the vertex $j \in [n]$. 
The arc set of $\D_i$ is
\[
\E_i = \big\{(\omega_{i, j}, i): \ j \in N^{\G}(i)\big\}. 
\]
Let $N^{\G}(i) = \{i_1,i_2,\ldots,i_{n_i}\}$. 
Then $\D_i$ can be drawn as in Fig.~\ref{fig:vertex_gadget}.

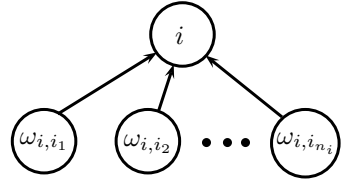
\begin{figure}[htb]
\centering{
\scalebox{1} 
{
\begin{pspicture}(0,-1.1792967)(6.3628125,1.1792967)
\pscircle[linewidth=0.04,dimen=outer](3.07625,0.73695296){0.4423438}
\psdots[dotsize=0.12](3.3939066,-0.70070326)
\psdots[dotsize=0.12](3.6739056,-0.70070326)
\psdots[dotsize=0.12](3.913906,-0.70070326)
\usefont{T1}{ptm}{m}{n}
\rput(3.0423436,0.72929674){$i$}
\usefont{T1}{ptm}{m}{n}
\rput(1.2423438,-0.6907033){$\omega_{i,i_1}$}
\usefont{T1}{ptm}{m}{n}
\rput(2.6223438,-0.73070323){$\omega_{i,i_2}$}
\usefont{T1}{ptm}{m}{n}
\rput(4.7323437,-0.71070325){$\omega_{i,i_{n_i}}$}
\psline[linewidth=0.04cm,arrowsize=0.05291667cm 2.0,arrowlength=1.4,arrowinset=0.4]{->}(1.4185936,-0.32070324)(2.74,0.49929675)
\psline[linewidth=0.04cm,arrowsize=0.05291667cm 2.0,arrowlength=1.4,arrowinset=0.4]{->}(2.7585938,-0.28070325)(2.96,0.33929676)
\psline[linewidth=0.04cm,arrowsize=0.05291667cm 2.0,arrowlength=1.4,arrowinset=0.4]{->}(4.4,-0.36070326)(3.38,0.45929676)
\pscircle[linewidth=0.04,dimen=outer](2.6162498,-0.68304706){0.4423438}
\pscircle[linewidth=0.04,dimen=outer](4.7092967,-0.71){0.46929675}
\pscircle[linewidth=0.04,dimen=outer](1.23625,-0.68304706){0.4423438}
\end{pspicture} 
}
}
\caption{Gadget $\D_i$ for each vertex $i$ of $\G$}
\label{fig:vertex_gadget}
\end{figure}
\vskip 10pt 

Additionally, we also introduce $n$ new vertices $p_1,p_2,\ldots,p_n$. The digraph $\D = (\VD, \ED)$ is built as follows. The vertex set of $\D$ is
\[ 
\VD = \big( \cup_{i=1}^n \V_i \big) 
\cup \{p_1,p_2,\ldots,p_n\}.
\] 
Let
\[
\Q_i = \big\{(p_i,i)\big\} \cup \big\{(p_i, \omega_{i', i}): \ i' \in [n],\ i \in N^{\G}(i') \big\}
\]
be the set consisting of $(p_i,i)$ and the arcs that connect $p_i$ and all the clones $\omega_{i', i}$ of $i$. The arc set of $\D$ is then defined to be
\[
\ED = \big( \cup_{i = 1}^n \E_i \big) 
\cup \big( \cup_{i = 1}^n \Q_i \big). 
\]  

\begin{figure}[htb]
\centering{
\scalebox{1} 
{
\begin{pspicture}(0,-1.0023438)(3.1046875,1.0023438)
\pscircle[linewidth=0.04,dimen=outer](2.6623437,0.56){0.4423438}
\pscircle[linewidth=0.04,dimen=outer](0.4423438,0.02){0.4423438}
\pscircle[linewidth=0.04,dimen=outer](2.6423438,-0.56){0.4423438}
\usefont{T1}{ptm}{m}{n}
\rput(0.4189063,0.0023438){$1$}
\usefont{T1}{ptm}{m}{n}
\rput(2.638906,0.5623438){$2$}
\usefont{T1}{ptm}{m}{n}
\rput(2.698906,-0.5776562){$3$}
\psline[linewidth=0.04cm](0.84,0.1423438)(2.26,0.4823438)
\psline[linewidth=0.04cm](0.8646876,-0.1076562)(2.22,-0.4576562)
\end{pspicture} 
}
}
\caption{An example of the graph $\G$}
\label{fig:G}
\end{figure}
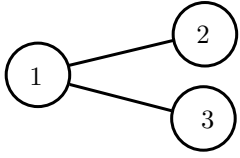
\vskip 10pt 

For example, if $\G$ is the graph in Fig.~\ref{fig:G}, then $\D$ is the digraph
in Fig.~\ref{fig:D}.

\vskip 10pt 
\begin{figure}[H]
\centering{
\scalebox{1} 
{
\begin{pspicture}(0,-2.1172173)(8.822812,2.4874701)
\pscircle[linewidth=0.04,dimen=outer](2.2052445,2.0451264){0.4423438}
\usefont{T1}{ptm}{m}{n}
\rput(2.2004008,2.02747){$1$}
\usefont{T1}{ptm}{m}{n}
\rput(1.1623437,0.24810487){$\omega_{1,2}$}
\usefont{T1}{ptm}{m}{n}
\rput(5.3423448,0.26810488){$\omega_{2,1}$}
\usefont{T1}{ptm}{m}{n}
\rput(3.102344,0.24810487){$\omega_{1,3}$}
\usefont{T1}{ptm}{m}{n}
\rput(7.582343,0.26810488){$\omega_{3,1}$}
\usefont{T1}{ptm}{m}{n}
\rput(2.1523442,-1.6518952){$p_1$}
\usefont{T1}{ptm}{m}{n}
\rput(5.3123446,-1.6118952){$p_2$}
\usefont{T1}{ptm}{m}{n}
\rput(7.5123444,-1.6118952){$p_3$}
\psline[linewidth=0.04cm,arrowsize=0.05291667cm 2.0,arrowlength=1.4,arrowinset=0.4]{->}(1.34,0.6381048)(2.0953126,1.6581048)
\psline[linewidth=0.04cm,arrowsize=0.05291667cm 2.0,arrowlength=1.4,arrowinset=0.4]{->}(2.9953127,0.6781049)(2.3553126,1.6781049)
\psline[linewidth=0.04cm,arrowsize=0.05291667cm 2.0,arrowlength=1.4,arrowinset=0.4]{->}(5.36,0.7181048)(5.355313,1.6381049)
\psline[linewidth=0.04cm,arrowsize=0.05291667cm 2.0,arrowlength=1.4,arrowinset=0.4]{->}(7.595313,0.7181049)(7.615313,1.6581048)
\rput{-85.47715}(5.156191,5.580168){\psarc[linewidth=0.04,arrowsize=0.05291667cm 2.0,arrowlength=1.4,arrowinset=0.4]{->}(5.5975986,0.0){2.8017166}{52.769493}{122.71867}}
\psline[linewidth=0.04cm,arrowsize=0.05291667cm 2.0,arrowlength=1.4,arrowinset=0.4]{->}(7.14,-1.4818952)(3.4353127,0.038104873)
\psline[linewidth=0.04cm,arrowsize=0.05291667cm 2.0,arrowlength=1.4,arrowinset=0.4]{->}(4.98,-1.4618952)(1.5353128,0.058104884)
\rput{-85.47715}(3.2998106,3.8335574){\psarc[linewidth=0.04,arrowsize=0.05291667cm 2.0,arrowlength=1.4,arrowinset=0.4]{->}(3.7242947,0.13120697){2.4633672}{48.392315}{125.42313}}
\psline[linewidth=0.04cm,arrowsize=0.05291667cm 2.0,arrowlength=1.4,arrowinset=0.4]{->}(2.18,-1.2418952)(2.215313,1.6581048)
\psline[linewidth=0.04cm,arrowsize=0.05291667cm 2.0,arrowlength=1.4,arrowinset=0.4]{->}(2.48,-1.3818952)(5.0153127,0.05810486)
\psline[linewidth=0.04cm,arrowsize=0.05291667cm 2.0,arrowlength=1.4,arrowinset=0.4]{->}(2.58,-1.6618952)(7.255313,0.038104903)
\pscircle[linewidth=0.04,dimen=outer](5.3652444,2.0451264){0.4423438}
\usefont{T1}{ptm}{m}{n}
\rput(5.3604007,2.02747){$2$}
\pscircle[linewidth=0.04,dimen=outer](7.5852447,2.0451264){0.4423438}
\usefont{T1}{ptm}{m}{n}
\rput(7.580401,2.02747){$3$}
\pscircle[linewidth=0.04,dimen=outer](1.1652445,0.2451264){0.4423438}
\pscircle[linewidth=0.04,dimen=outer](3.0852444,0.2451264){0.4423438}
\pscircle[linewidth=0.04,dimen=outer](5.3452444,0.2651264){0.4423438}
\pscircle[linewidth=0.04,dimen=outer](7.5852447,0.2651264){0.4423438}
\pscircle[linewidth=0.04,dimen=outer](2.1652446,-1.6748736){0.4423438}
\pscircle[linewidth=0.04,dimen=outer](5.3452444,-1.6148736){0.4423438}
\pscircle[linewidth=0.04,dimen=outer](7.5452447,-1.6148736){0.4423438}
\end{pspicture} 
}
}
\caption{The digraph $\D$ built from the graph $\G$ in Fig.~\ref{fig:G}}
\label{fig:D}
\end{figure}
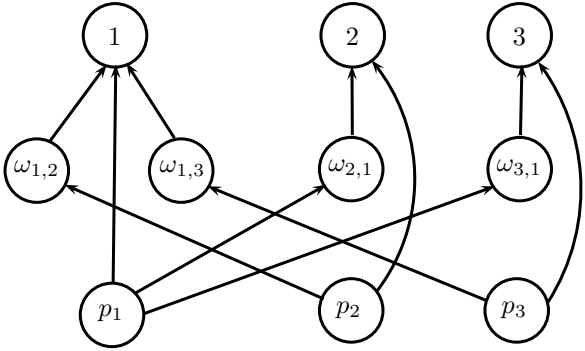
\vskip 10pt 

Our goal now is to show that $\G$ is $k$-colorable if and only if 
$\D$ is fairly $k$-colorable. 

Suppose that $\G$ is $k$-colorable and $\phi_{\G}: [n] \ra [k]$ is a $k$-coloring of $\G$. We consider the mapping $\phi_{\D}: \VD \ra [k]$ defined as follows

\begin{enumerate}
	\item For every $i \in [n]$, $\phi_{\D}(i)  \define \phi_{\G}(i)$;
	\item If $i \in N^{\G}(i')$ then $\phi_{\D}(\omega_{i',i}) \define \phi_{\D}(i) = \phi_{\G}(i)$, in other words, clones of $i$ have the same color as $i$;  
  \item For every $i \in [n]$, $\phi_{\D}(p_i)$ can be chosen arbitrarily, as long as it is different from $\phi_{\D}(i)$. 
\end{enumerate}

We claim that $\phi_{\D}$ is a fair $k$-coloring for $\D$. We first verify the condition (C1) (see Definition~\ref{def:fair_coloring}). It is straightforward from the definition of $\phi_{\D}$ that the endpoints of each of the arcs of the forms $(p_i,i)$ for $i \in [n]$, and $(p_i, \omega_{i', i})$ for $i \in N^{\G}(i')$, have different colors. It remains to check if $i$ and $\omega_{i,j}$ for $j \in N^{\G}(i)$ have different colors. On the one hand, $\omega_{i,j}$ is a clone of $j$, and hence has the same color as $j$. In other words, 
\[
\phi_{\D}(\omega_{i, j}) = \phi_{\D}(j) = \phi_{\G}(j).
\] 
On the other hand, since $j \in N^{\G}(i)$, we obtain that
\[
\phi_{\G}(j) \neq \phi_{\G}(i) = \phi_D(i).
\]
Therefore, $\phi_{\D}(\omega_{i,j}) \neq \phi_{\D}(i)$ for all $i \in [n]$ and $j \in N^{\G}(i)$. 
Thus, (C1) is satisfied.  
 
We now check if (C2) (see Definition~\ref{def:fair_coloring}) is also satisfied. 
The out-neighbors of $p_i$ are $i$ and its clones $\omega_{i',i}$ ($i \in N^{\G}(i')$). 
These vertices have the same color in $\D$, namely $\phi_{\G}(i)$, by the definition of $\phi_{\D}$.
Thus (C2) is also satisfied. Therefore $\phi_{\D}$ is a fair $k$-coloring of $\D$.  

Conversely, suppose that $\phi_{\D}: \VD \ra [k]$ is a fair $k$-coloring of $\D$. 
Condition (C2) guarantees that all clones of $i$ have the same color as $i$, namely, $\phi_{\D}(\omega_{i',i}) = \phi_{\D}(i)$ if $i \in N^{\G}(i')$. 
Therefore, by (C1), if $\{i,j\} \in \EG$, that is, $j \in N^{\G}(i)$, then 
\[
\phi_{\D}(i) \neq \phi_{\D}(\omega_{i,j}) = \phi_{\D}(j). 
\]
Hence, if we define $\phi_{\G}: [n] \ra [k]$ by $\phi_{\G}(i) = \phi_{\D}(i)$ for all $i \in [n]$, then it is a $k$-coloring of $\G$. Thus $\G$ is $k$-colorable. 

Finally, notice that the order of $\D$ is a polynomial with respect to the order of $\G$. More specifically, $|\VD| = 2|\VG| + 2|\EG|$ and $|\ED| = |\VG| + 4|\EG|$. Moreover, building $\D$ from $\G$, and also obtaining a coloring of $\G$ from a coloring of $\D$, can be done in polynomial time with respect to the order of $\G$. 
Since the $k$-coloring problem ($k \geq 3$) is NP-hard \cite{Karp1972}, 
we conclude that the fair $k$-coloring problem is also NP-hard. 
\end{proof} 
\vskip 10pt 

According to Theorem~\ref{thm:fair_coloring_NP} and the work by Blasiak 
{\et}~\cite{BlasiakKleinbergLubetzky2011} (see the discussion after 
Corollary~\ref{coro:mr_color}), we obtain the following. 

\vskip 10pt 
\begin{theorem}
\label{thm:NP_complete}
Let $\D$ be an arbitrary digraph. Then the decision problem whether $\mrt(\D) = 2$ is NP-complete. 
However, the decision problem whether $\bt_2(\D) = 2$ can be solved in polynomial time. 
\end{theorem} 
\vskip 10pt 

Recall that by contrast, for a graph $\G$, it was observed by Peeters~\cite{Peeters96} that $\G$ has min-rank two if and only if $\Gc$ is a bipartite graph and $\G$ is not a complete graph, which can be verified in polynomial time
(see, for instance, West~\cite[p. 495]{West}). Note that a graph is bipartite if and only if it is $2$-colorable.  
This fact can also be derived by applying Theorem~\ref{thm:mr_color} to the digraph obtained from $\G$ by replacing each edge of $\G$ by two arcs of opposite directions. 

\section{Circuit-Packing Bound}
\label{sec:circuit_packing_bound}

In this section we introduce a new upper bound for the min-rank of 
a digraph. This bound reveals some new families of digraphs whose min-ranks are computable in polynomial time. 

\subsection{The Bound}
\label{subsec:circuit_packing_bound}

Let $\nD$ be the \emph{circuit packing number} of $\D$, namely, the maximum number of 
vertex-disjoint circuits in $\D$. Below, we establish an upper bound on min-ranks of digraphs,
which uses the circuit packing number. This bound was first presented by Chaudhry {\et} in~\cite{ChaudhryAsadSprintsonLangberg2011}, and was obtained independently by the authors of this paper approximately at the same time.

\vskip 5pt 
\begin{proposition}[Circuit-packing bound]       
\label{pro:circuit_packing_bound}
The following holds for every digraph $\D$ of order $n$:
\[
\mrq(\D) \leq n - \nD. 
\]
\end{proposition} 
\begin{proof}
Suppose $\D$ contains $\nD$ vertex-disjoint circuits $\C_1, \C_2, \ldots, \C_{\nD}$, 
where 
\[
\C_i = \big(u_{i,1}, u_{i,2}, \ldots, u_{i, n_i}\big), \ i \in [\nD],\ 2 \leq n_i \leq n. 
\]
Let $\V(\C_i) = \{u_{i,1}, u_{i,2}, \ldots, u_{i, n_i}\}$ ($i \in [\nD]$). 
We construct a matrix $\bM$ fitting $\D$ as follows. 
Let 
\[
\A \define \VD \setminus \cup_{i \in [\nD]}\V(\C_i).
\]
For $v \in \A$ let $\bM_v = \be_v$. 
For $i \in [\nD]$ and $s \in [n_i - 1]$, let 
\[
\bM_{u_{i, s}} = \be_{u_{i,s}} - \be_{u_{i, s+1}},
\] 
and let 
\[
\bM_{u_{i, n_i}} = \be_{u_{i, 1}} - \be_{u_{i, n_i}}.
\] 
Clearly, $\bM$ fits $\D$. 
Moreover, as 
\[
\bM_{u_{i, n_i}} = \sum_{s = 1}^{n_i-1} \bM_{u_{i , s}},
\] 
we have
\[
\rank\left(\bM_{\V(\C_i)}\right) \leq n_i - 1
\]
for all $i \in [\nD]$. 
Since $\V(\C_i)$'s, $i \in [\nD]$, are pairwise disjoint, we have
\[
\begin{split} 
\rank(\bM) &\leq \sum_{i = 1}^{\nD} \rank\left(\bM_{\V(\C_i)}\right) + \rank\left(\bM_\A\right)\\
&\leq  \sum_{i = 1}^{\nD} (n_i - 1) + \left(n -  \sum_{i = 1}^{\nD} n_i\right)\\
&= n - \nD. 
\end{split} 
\]
Thus, $\mrq(\D) \leq n - \nD$. 
\end{proof} 
\vskip 10pt 

Whereas for graphs the clique-cover bound is the best known bound, for digraphs that are not symmetric, this is not the case. The worst scenario for the clique-cover bound is when the digraph has no two arcs of opposite directions. For such a digraph, this bound becomes trivial, as the size of the smallest clique cover is equal to the order of the digraph.  
The following example emphasizes the fact that for certain digraphs, the circuit-packing bound can be \emph{significantly tighter} than the clique-cover bound. 

\vskip 10pt 
\begin{example}
\label{ex:comparing_bounds}
Let $\D_k$ be the digraph of order $n = 3k$ depicted in Fig.~\ref{fig:comparing_bounds}. 
As there are no arcs of opposite directions, all cliques in $\D_k$ are of cardinality one. Therefore, the clique-cover 
bound gives $\mrq(\D_k) \leq 3k$. On the other hand, as $\D_k$ contains $k$ vertex-disjoint 
circuits, namely $\C_i = (3i+1, 3i + 2, 3i+3)$ for $i = 0,1,\ldots,k-1$, the circuit-packing bound
yields $\mrq(\D_k) \leq 2k = 3k - k$. The gap between the two bounds is one third of the order of the digraph.  
\end{example} 
 
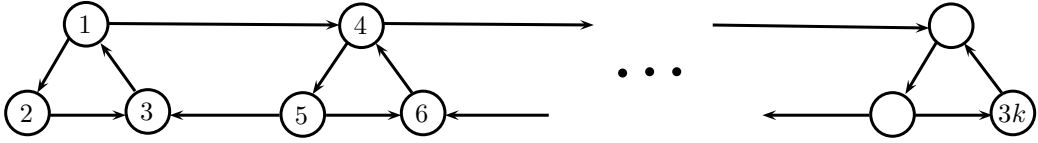
\begin{figure*}[htb]
\centering{
\scalebox{1} 
{
\begin{pspicture}(0,-0.91)(13.801875,0.91)
\pscircle[linewidth=0.04,dimen=outer](1.09,0.6){0.31}
\usefont{T1}{ptm}{m}{n}
\rput(1.0728124,0.6){$1$}
\pscircle[linewidth=0.04,dimen=outer](0.31,-0.58){0.31}
\usefont{T1}{ptm}{m}{n}
\rput(0.2928125,-0.58){$2$}
\pscircle[linewidth=0.04,dimen=outer](1.91,-0.56){0.31}
\usefont{T1}{ptm}{m}{n}
\rput(1.8928125,-0.56){$3$}
\pscircle[linewidth=0.04,dimen=outer](4.75,0.58){0.31}
\usefont{T1}{ptm}{m}{n}
\rput(4.7328124,0.58){$4$}
\pscircle[linewidth=0.04,dimen=outer](3.97,-0.6){0.31}
\usefont{T1}{ptm}{m}{n}
\rput(3.9528124,-0.62){$5$}
\pscircle[linewidth=0.04,dimen=outer](5.57,-0.58){0.31}
\usefont{T1}{ptm}{m}{n}
\rput(5.5528126,-0.58){$6$}
\psdots[dotsize=0.12](8.22,-0.05)
\psdots[dotsize=0.12](8.58,-0.03)
\psdots[dotsize=0.12](8.94,-0.03)
\pscircle[linewidth=0.04,dimen=outer](12.59,0.58){0.31}
\pscircle[linewidth=0.04,dimen=outer](11.81,-0.6){0.31}
\pscircle[linewidth=0.04,dimen=outer](13.41,-0.58){0.31}
\psline[linewidth=0.04cm,arrowsize=0.05291667cm 2.0,arrowlength=1.4,arrowinset=0.4]{->}(1.38,0.61)(4.46,0.59)
\psline[linewidth=0.04cm,arrowsize=0.05291667cm 2.0,arrowlength=1.4,arrowinset=0.4]{->}(5.04,0.61)(7.84,0.59)
\psline[linewidth=0.04cm,arrowsize=0.05291667cm 2.0,arrowlength=1.4,arrowinset=0.4]{->}(9.42,0.59)(12.3,0.55)
\psline[linewidth=0.04cm,arrowsize=0.05291667cm 2.0,arrowlength=1.4,arrowinset=0.4]{->}(1.74,-0.29)(1.28,0.35)
\psline[linewidth=0.04cm,arrowsize=0.05291667cm 2.0,arrowlength=1.4,arrowinset=0.4]{->}(0.86,0.41)(0.44,-0.31)
\psline[linewidth=0.04cm,arrowsize=0.05291667cm 2.0,arrowlength=1.4,arrowinset=0.4]{->}(0.6,-0.61)(1.62,-0.61)
\psline[linewidth=0.04cm,arrowsize=0.05291667cm 2.0,arrowlength=1.4,arrowinset=0.4]{->}(3.66,-0.61)(2.2,-0.61)
\psline[linewidth=0.04cm,arrowsize=0.05291667cm 2.0,arrowlength=1.4,arrowinset=0.4]{->}(5.42,-0.31)(4.94,0.35)
\psline[linewidth=0.04cm,arrowsize=0.05291667cm 2.0,arrowlength=1.4,arrowinset=0.4]{->}(4.56,0.35)(4.1,-0.33)
\psline[linewidth=0.04cm,arrowsize=0.05291667cm 2.0,arrowlength=1.4,arrowinset=0.4]{->}(4.28,-0.61)(5.28,-0.61)
\psline[linewidth=0.04cm,arrowsize=0.05291667cm 2.0,arrowlength=1.4,arrowinset=0.4]{->}(7.24,-0.61)(5.86,-0.61)
\psline[linewidth=0.04cm,arrowsize=0.05291667cm 2.0,arrowlength=1.4,arrowinset=0.4]{->}(11.5,-0.61)(10.08,-0.61)
\psline[linewidth=0.04cm,arrowsize=0.05291667cm 2.0,arrowlength=1.4,arrowinset=0.4]{->}(13.26,-0.31)(12.76,0.35)
\psline[linewidth=0.04cm,arrowsize=0.05291667cm 2.0,arrowlength=1.4,arrowinset=0.4]{->}(12.4,0.37)(11.98,-0.33)
\psline[linewidth=0.04cm,arrowsize=0.05291667cm 2.0,arrowlength=1.4,arrowinset=0.4]{->}(12.1,-0.61)(13.14,-0.61)
\usefont{T1}{ptm}{m}{n}
\rput(13.401406,-0.58){$3k$}
\end{pspicture} 
}
}
\caption{Example where the circuit-packing bound is tighter than the clique-cover bound}
\label{fig:comparing_bounds}
\end{figure*} 

\subsection{Digraphs Attaining Circuit-Packing Bound}
\label{subsec:graphs_attaining_bound}

In this subsection, we present several new examples of families of digraphs 
that attain the circuit-packing bound. 

A \emph{feedback vertex (arc, respectively) set} of $\D$ is a set of vertices (arcs, respectively) whose removal destroys all
circuits in $\D$. Let $\tD$ ($\teD$, respectively) denote the \emph{minimum size} of a feedback vertex (arc, respectively) set of $\D$.\index{feedback set!vertex}\index{feedback set!arc}\index{$\tD$}\index{$\teD$}  
Then it is clear that $\aD = n - \tD$. 

\vskip 10pt
\begin{corollary}
\label{coro:bounds_coincide}
If $\nD = \tD$ then 
\begin{equation}
\label{equ:sandwich}
\mrq(\D) = n - \nD = n - \tD.
\end{equation}
\end{corollary}
\begin{proof}
By Corollary~\ref{coro:aD} and
Proposition~\ref{pro:circuit_packing_bound} we have
\[
n - \tD \leq \mrq(\D) \leq n - \nD. 
\]
Hence, the proof follows. 
\end{proof} 
\vskip 10pt

When $\D$ satisfies $\nD = \tD$, we say that $\D$ satisfies the \emph{min-max vertex equality}. 
In that case, the circuit-packing bound is tight. 
Similarly, let $\neD$ denote the maximum number of arc-disjoint circuits in $\D$. \index{$\neD$}
We say that $\D$ satisfies the \emph{min-max arc equality} if $\neD = \teD$. 

The  first example of digraphs for which the circuit-packing bound is tight is the \emph{fully reducible flow digraphs} \cite{FrankGyarfas1976}. A flow digraph is a digraph that contains a special vertex called root, from which any vertex is reachable by a directed path. A fully reducible flow digraph is a flow digraph that satisfies the property that every circuit $\C$ in the digraph has a unique vertex $v_\C$ such that every directed path from the root to a vertex of $\C$ must contain $v_\C$.
Interestingly, it was proved by Shamir \cite{Shamir1979} that there is 
a \emph{linear time} algorithm to find $\nD$ ($= \tD$) for a fully reducible flow digraph $\D$. As a consequence, the min-rank of a fully reducible flow digraph (recognizable in polynomial time with respect to its size \cite{Tarjan1974}) can be calculated in linear time with respect to its size.\index{digraph!fully reducible flow}

The second example of digraphs that satisfy the min-max vertex equality is the \emph{connectively reducible digraphs} \cite{Szwarcfiter1989}. This family of digraphs actually generalizes both the family of fully reducible flow digraphs and the family of \emph{cyclically reducible digraphs} \cite{Wang_Lloyd_Soffa1985}. A polynomial time algorithm was provided by Szwarcfiter~\cite{Szwarcfiter1989} to recognize a member of this family and subsequently find a maximum set of vertex-disjoint circuits as well as a minimum feedback vertex set.\index{digraph!connectively reducible} 
Therefore, by Corollary~\ref{coro:bounds_coincide}, 
(\ref{equ:sandwich}) holds
for a connectively reducible digraph $D$. Moreover, 
$\mrq(\D)$ can be found in polynomial time. 
 
The third example of digraphs for which the circuit-packing bound is tight is the digraphs that \emph{pack} \cite{GueninThomas2001}. \index{digraph!that pack}
A digraph packs if the min-max vertex equality holds for all of its subgraphs. The digraphs in this family are exactly ones that have no minor isomorphic to an odd double circuit or $F_7$, a special digraph of order $7$ (interested readers may refer to \cite{GueninThomas2001} for 
more details, also for a structural characterization of this family of digraphs). For instance, \emph{strongly planar} digraphs
\cite{GueninThomas2001} belong to this family. As far as we know, there are no known polynomial time algorithms to find a minimum feedback vertex set of a digraph that packs.  

The other examples of digraphs for which the circuit-packing bound is tight are the 
\emph{line digraphs} of planar digraphs, of fully reducible flow digraphs, and of (special) Eulerian digraphs~\cite{Seymour1996}. 

\vskip 10pt
\begin{definition}  
Let $\D = (\VD, \ED)$ be a digraph. Then the digraph $\LL = (\V(\LL), \E(\LL))$ 
with $\V(\LL) = \ED$ and
\[
\E(\LL) = \big\{(e,e'):\ e = (u,v) \in \ED,\ e' = (v,w) \in \ED \big\},
\]
is called
the \emph{line digraph} of $\D$. We denote the line digraph of $\D$ by $\LD$.
The digraph $\D$ is called a \emph{root digraph} of $\LD$.
\end{definition} 

\begin{figure}
\centering
\subfloat[A digraph $\D$]{
\scalebox{1} 
{
\begin{pspicture}(0,-1.7137657)(2.890088,1.8062344)
\pscircle[linewidth=0.04,dimen=outer](1.4927025,1.4962343){0.31}
\usefont{T1}{ptm}{m}{n}
\rput(1.4741088,1.4962343){$1$}
\pscircle[linewidth=0.04,dimen=outer](1.4727025,0.05623431){0.31}
\usefont{T1}{ptm}{m}{n}
\rput(1.4541088,0.05623431){$2$}
\pscircle[linewidth=0.04,dimen=outer](1.4727025,-1.4037657){0.31}
\usefont{T1}{ptm}{m}{n}
\rput(1.4541088,-1.4037657){$3$}
\rput{54.15276}(1.3342782,-1.0928671){\psarc[linewidth=0.04,arrowsize=0.05291667cm 2.0,arrowlength=1.4,arrowinset=0.4]{->}(1.7360504,0.75859696){0.8872246}{71.56505}{180.0}}
\rput{54.15276}(0.14254215,-1.6816413){\psarc[linewidth=0.04,arrowsize=0.05291667cm 2.0,arrowlength=1.4,arrowinset=0.4]{->}(1.7160504,-0.7014031){0.8872246}{71.56505}{180.0}}
\rput{234.89763}(2.5697513,-0.1879949){\psarc[linewidth=0.04,arrowsize=0.05291667cm 2.0,arrowlength=1.4,arrowinset=0.4]{->}(1.2360504,-0.7614031){0.8872246}{71.56505}{180.0}}
\rput{234.12204}(1.377506,2.2209172){\psarc[linewidth=0.04,arrowsize=0.05291667cm 2.0,arrowlength=1.4,arrowinset=0.4]{->}(1.2560503,0.75859696){0.8872246}{71.56505}{180.0}}
\usefont{T1}{ptm}{m}{n}
\rput{-0.5897514}(-0.008163664,0.0068681664){\rput(0.6441088,0.8162343){$e_1$}}
\usefont{T1}{ptm}{m}{n}
\rput{-2.3030555}(0.02959154,0.025261521){\rput(0.6241088,-0.7037657){$e_2$}}
\usefont{T1}{ptm}{m}{n}
\rput{-1.602764}(0.02118291,0.06749273){\rput(2.4041088,-0.7037657){$e_3$}}
\usefont{T1}{ptm}{m}{n}
\rput{358.25943}(-0.023076262,0.0739689){\rput(2.4041088,0.8162343){$e_4$}}
\end{pspicture} 
}
}
\subfloat[The line digraph $\LD$]{
\scalebox{1} 
{
\begin{pspicture}(0,-1.1)(2.5628126,1.1)
\pscircle[linewidth=0.04,dimen=outer](0.42953125,0.79){0.31}
\usefont{T1}{ptm}{m}{n}
\rput(0.40234375,0.79){$e_1$}
\pscircle[linewidth=0.04,dimen=outer](2.0695312,0.79){0.31}
\usefont{T1}{ptm}{m}{n}
\rput(2.0423439,0.79){$e_2$}
\pscircle[linewidth=0.04,dimen=outer](2.0895312,-0.79){0.31}
\usefont{T1}{ptm}{m}{n}
\rput(2.0823438,-0.79){$e_3$}
\pscircle[linewidth=0.04,dimen=outer](0.44953126,-0.77){0.31}
\usefont{T1}{ptm}{m}{n}
\rput(0.42234376,-0.77){$e_4$}
\psline[linewidth=0.04cm,arrowsize=0.05291667cm 2.0,arrowlength=1.4,arrowinset=0.4]{->}(0.71953124,0.82)(1.7995313,0.78)
\psline[linewidth=0.04cm,arrowsize=0.05291667cm 2.0,arrowlength=1.4,arrowinset=0.4]{->}(2.0795312,0.5)(2.0795312,-0.52)
\psline[linewidth=0.04cm,arrowsize=0.05291667cm 2.0,arrowlength=1.4,arrowinset=0.4]{->}(1.7995313,-0.8)(0.75953126,-0.8)
\psline[linewidth=0.04cm,arrowsize=0.05291667cm 2.0,arrowlength=1.4,arrowinset=0.4]{->}(0.45953125,-0.46)(0.45953125,0.5)
\psline[linewidth=0.04cm,arrowsize=0.05291667cm 2.0,arrowlength=1.4,arrowinset=0.4]{->}(0.37953115,0.5)(0.37953115,-0.52)
\psline[linewidth=0.04cm,arrowsize=0.05291667cm 2.0,arrowlength=1.4,arrowinset=0.4]{->}(1.9995313,-0.48)(1.9995313,0.48)
\end{pspicture} 
}
}
\caption{Example of a digraph and its line digraph}
\label{fig:line_digraph}
\end{figure}
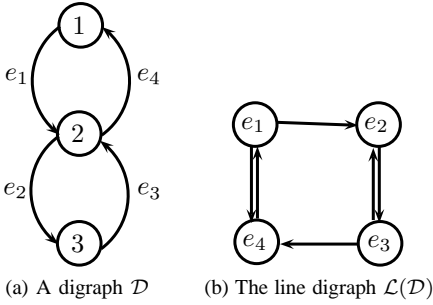 

\begin{lemma} 
\label{lem:first_equality}
$\nu_0(\LD) = \neD$.
\end{lemma}
\begin{proof} 
\mbox{}
\begin{enumerate}
	\item $\nu_0(\LD) \geq \neD$. 
	It suffices to show that the existence of a set of arc-disjoint circuits in $\D$
	implies the existence of a set of vertex-disjoint circuits of the same size in $\LD$. 
	Let $\{\C_1, \C_2, \ldots, \C_k\}$ be a set of arc-disjoint circuits in $\D$, 
	where $\C_i = (v_{i,1}, v_{i,2}, \ldots, v_{i, r_i})$, $r_i \geq 2$, $i \in [k]$.
	Let $e_{i,j} = (v_{i,j}, v_{i,j+1})$, for $i \in [k]$ and $j \in [r_i - 1]$. 
	Moreover, let $e_{i, r_i} = (v_{i,r_i},v_{i,1})$ for $i \in [k]$.
	Let $\C'_i = (e_{i,1}, e_{i,2}, \ldots, e_{i,r_i})$ for $i \in [k]$. 
	Then $\C'_i$ is also a circuit in $\LD$ for every $i \in [k]$.
	Moreover, as the circuits $\C_1, \C_2, \ldots, \C_k$ share no common edges in $\D$, 
	we deduce that $\C'_1, \C'_2, \ldots, \C'_k$ share no common vertices in $\LD$. 
	Therefore, they form a set of $k$ vertex-disjoint circuits in $\LD$. 

\item $\nu_0(\LD) \leq \neD$. 
	It suffices to show that the existence of a set of vertex-disjoint circuits in $\LD$
	implies the existence of a set of arc-disjoint circuits of the same size in $\D$. 
	Let $\{\C'_1, \C'_2, \ldots, \C'_k\}$ be a set of vertex-disjoint circuits in $\LD$, 
	where $\C'_i = \{e_{i,1}, e_{i,2}, \ldots, e_{i, r_i}\}$ for $i \in [k]$. 
	Suppose that $e_{i,j} = (v_{i,j}, v_{i,j+1}) \in \ED$ for $i \in [k]$ and $j \in [r_i]$, 
	where $v_{i,j}$ and $v_{i,j+1}$ are vertices of $\D$. 
	Then $v_{i, r_i + 1} \equiv v_{i,1}$ for $i \in [k]$. 
	For each $i \in [k]$, consider the sequence of (possibly repeated) vertices
	\[
	v_{i,1}, v_{i,2}, \ldots, v_{i, r_i + 1}.
	\]
	Since $v_{i,1} \equiv v_{i,r_i + 1}$ and $(v_{i,j}, v_{i,j+1}) \in \ED$ for all 
	$j \in [r_i]$, there exist $j_0$ and $j_1$ such that
\begin{itemize}
	\item
	$1 \leq j_0 < j_1 \leq r_i$; 
	\item 
	$v_{i,j_0} \equiv v_{i,j_1 + 1}$; 
	\item $v_{i,j_0}, v_{i,j_0 + 1}, \ldots, v_{i, j_1}$ are distinct. 
\end{itemize}
Then $\C_i = (v_{i,j_0}, v_{i,j_0 + 1}, \ldots, v_{i, j_1})$
	is a circuit in $\D$. Since the circuits $\C'_1, \C'_2, \ldots, \C'_k$ share 
	no common vertices in $\LD$, we obtain that the circuits $\C_1,\C_2,\ldots,\C_k$
	share no common edges in $\D$. \qedhere
\end{enumerate}
\end{proof} 
\vskip 10pt 

\begin{lemma}
\label{lem:second_equality}
$\tau_0(\LD) = \teD$. 
\end{lemma}
\begin{proof} 
	Let $F = \{e_1,e_2,\ldots, e_k\}$, where $e_i \in \ED$ for $i \in [k]$, 
	be an arbitrary set of arcs of $\D$. We can also view $F$ as a set of vertices
	of $\LD$. It suffices to show that $F$ is a feedback arc set of $\D$ if and only if 
	$F$ is a feedback vertex set of $\LD$, for every such set $F$. 
	
	Let $\D - F$ be the digraph obtained from $\D$ by removing
	all arcs in $F$. Let $\LD - F$ be the digraph obtained from 
	$\LD$ by removing all vertices in $F$. Then $\LD - F = \LL(\D - F)$. 
	As shown in the proof of Lemma~\ref{lem:first_equality}, 
	the existence of a circuit in $\D - F$ would result in the existence
	of a circuit in $\LL(\D - F)$ and vice versa. 
	Therefore, $\D - F$ is acyclic if and only if $\LD - F$
	is acyclic. Thus, $F$ is a feedback arc set of $\D$ if and only if 
	$F$ is a feedback vertex set of $\LD$. \qedhere
\end{proof} 

\vskip 10pt 
\begin{proposition}
\label{pro:digraph_line-digraph}
Let $\D$ be a digraph. If $\neD = \teD$ then $\nu_0(\LD) = \tau_0(\LD)$ and 
\[
\mrq(\LD) = |\ED| - \neD.
\] 
\end{proposition}
\begin{proof} 
Suppose that $\neD = \teD$. By Lemma~\ref{lem:first_equality} and Lemma~\ref{lem:second_equality}, $\nu_0(\LD) = \tau_0(\LD)$. Therefore, by applying Corollary~\ref{coro:bounds_coincide}
to $\LD$ we obtain
\[
\mrq(\LD) = |\V(\LD)| - \nu_0(\LD) = |\ED| - \neD.\qedhere
\]
\end{proof} 

\begin{definition}       
A digraph that can be drawn on a plane in such a way that its
(arcs) edges intersect only at their endpoints is called \emph{planar}. 
\end{definition}
\vskip 10pt 

It is known that the min-max arc equality is satisfied for 
planar digraphs \cite{LucchesiYounger1978}, for fully reducible flow digraphs \cite{Ramachandran1990}, 
and for a special family of Eulerian digraphs \cite{Seymour1996}. Therefore, 
by Proposition~\ref{pro:digraph_line-digraph}, the min-max vertex equality is satisfied 
for the line digraphs of the members of these families. In summary, we have the following. 

\vskip 10pt 
\begin{corollary}
\label{coro:graphs_attaining_bound}
The circuit-packing bound is tight for the following families of digraphs: 
connectively reducible digraphs, digraphs that pack,
line digraphs of planar digraphs, line digraphs of fully reducible flow digraphs, 
and line digraphs of special Eulerian digraphs. 
\end{corollary}     
\vskip 10pt 

Consider the ICSI instances described by digraphs $\D$ with $\mrq(\D) = \aD$. 
By Theorem~\ref{thm:sandwiched_theorem_graph}, $\mrq(\D) = \bt_q(\D)$. 
Hence, for such instances, \emph{scalar linear} index codes are as good as 
\emph{vector nonlinear} index codes, in terms of transmission rates.
Thus, for the ICSI instances described by families of digraphs listed in
Corollary~\ref{coro:graphs_attaining_bound}, 
scalar linear index codes achieve the best possible transmission rates. 
Previously, only perfect graphs and acyclic digraphs were known to have this property
\cite{Yossef-journal}. 

\vskip 10pt 
\begin{definition}             
A digraph is called \emph{partially planar} if all of its strongly connected components
are planar. 
\end{definition} 
\vskip 10pt 

Since the strongly connected components of a planar digraph are also planar, 
a planar digraph is partially planar. However, the converse is not always true, 
as shown in Fig.~\ref{fig:partially_planar}. 

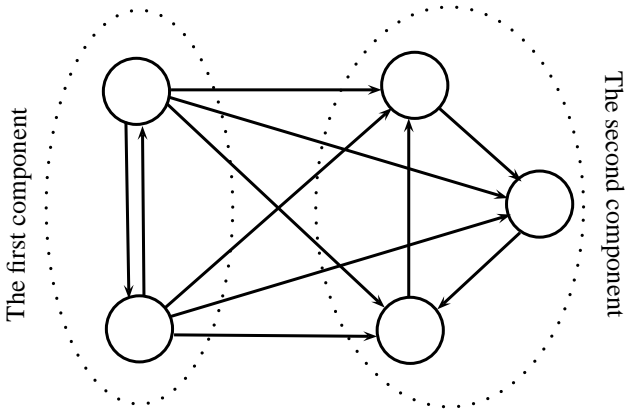
\begin{figure}[htb]
\centering{
\scalebox{1} 
{
\begin{pspicture}(0,-2.682)(8.240747,2.682)
\pscircle[linewidth=0.04,dimen=middle](5.4157376,1.618){0.44}
\pscircle[linewidth=0.04,dimen=middle](1.7557378,-1.622){0.44}
\pscircle[linewidth=0.04,dimen=middle](1.7157379,1.538){0.44}
\pscircle[linewidth=0.04,dimen=middle](7.075738,0.038){0.44}
\pscircle[linewidth=0.04,dimen=middle](5.3557377,-1.642){0.44}
\psline[linewidth=0.04cm,arrowsize=0.05291667cm 2.0,arrowlength=1.4,arrowinset=0.4]{->}(5.735738,1.318)(6.8357377,0.338)
\psline[linewidth=0.04cm,arrowsize=0.05291667cm 2.0,arrowlength=1.4,arrowinset=0.4]{->}(6.8157377,-0.342)(5.715738,-1.362)
\psline[linewidth=0.04cm,arrowsize=0.05291667cm 2.0,arrowlength=1.4,arrowinset=0.4]{->}(5.3357377,-1.202)(5.3357377,1.178)
\psline[linewidth=0.04cm,arrowsize=0.05291667cm 2.0,arrowlength=1.4,arrowinset=0.4]{->}(2.1757379,1.558)(4.9757376,1.558)
\psline[linewidth=0.04cm,arrowsize=0.05291667cm 2.0,arrowlength=1.4,arrowinset=0.4]{->}(2.1557379,1.458)(6.655738,0.118)
\psline[linewidth=0.04cm,arrowsize=0.05291667cm 2.0,arrowlength=1.4,arrowinset=0.4]{->}(2.1157377,1.378)(5.017565,-1.34)
\psline[linewidth=0.04cm,arrowsize=0.05291667cm 2.0,arrowlength=1.4,arrowinset=0.4]{->}(2.0957377,-1.342)(5.115738,1.318)
\psline[linewidth=0.04cm,arrowsize=0.05291667cm 2.0,arrowlength=1.4,arrowinset=0.4]{->}(2.1557379,-1.462)(6.695738,-0.102)
\psline[linewidth=0.04cm,arrowsize=0.05291667cm 2.0,arrowlength=1.4,arrowinset=0.4]{->}(2.2157378,-1.702)(4.9157376,-1.722)
\psellipse[linewidth=0.044,linestyle=dotted,dotsep=0.16cm,dimen=middle](5.879799,0.0)(1.78,2.66)
\psellipse[linewidth=0.044,linestyle=dotted,dotsep=0.16cm,dimen=middle](1.7557378,-0.0020000006)(1.24,2.62)
\usefont{T1}{ptm}{m}{n}
\rput{89.421364}(0.07857694,-0.27507597){\rput(0.16728503,-0.07441322){The first component}}
\usefont{T1}{ptm}{m}{n}
\rput{-89.37804}(7.7639565,8.184342){\rput(8.007876,0.19126037){The second component}}
\psline[linewidth=0.04cm,arrowsize=0.05291667cm 2.0,arrowlength=1.4,arrowinset=0.4]{->}(1.5757377,1.138)(1.6157378,-1.222)
\psline[linewidth=0.04cm,arrowsize=0.05291667cm 2.0,arrowlength=1.4,arrowinset=0.4]{->}(1.8157378,-1.182)(1.7957379,1.078)
\end{pspicture} 
}
}
\caption{A partially planar digraph that is not planar}
\label{fig:partially_planar}
\end{figure}

\vskip 10pt 
\begin{proposition}
\label{pro:line_digraph_planar}
There is a polynomial time algorithm to recognize the line digraph of a 
partially planar digraph and subsequently determine its min-rank.
\end{proposition}
\begin{proof}
\mbox{}
\begin{enumerate}
\item {\bf Recognition Phase:} \\
There is 
a one-to-one correspondence between the set of strongly connected components 
of order at least two of $\D$ and the set of strongly connected components
of $\LD$ in the following sense. If $\D_i$'s, $i \in [k]$, are all strongly connected components of $\D$
each of which contains at least two vertices, 
then $\LL(\D_i)$'s, $i \in [k]$, are all strongly connected components of $\LD$. 
Therefore, to determine whether 
a given digraph $\LL$ is the line digraph of a partially planar digraph, 
it suffices to determine whether each of its strongly connected components 
$\LL_i$ ($i \in [k]$) is the line digraph of a planar digraph.  
Note also that we can find all strongly connected components of a digraph 
in time linear in the number of edges \cite{Tarjan1972}. 

For each $i \in [k]$, employing a polynomial time algorithm, 
we can determine whether $\LL_i$ is a line digraph of a digraph \cite{Syslo1982}. 
If the answer is YES, then the algorithm also outputs a digraph $\D'_i$, which 
is a root digraph of $\LL_i$ and is strongly connected. 

Suppose $\LL = \LL(\D)$, where $\D$ is a digraph. 
Moreover, let $\LL_i = \LL(\D_i)$, where $\D_i$'s, $i \in [k]$, are all 
strongly connected components of $\D$ of order at least two. 
By \cite[Theorem 3]{HararyNorman1960}, $\D'_i$ and $\D_i$ are actually isomorphic, $i \in [k]$. 
Hence, to complete the Recognition Phase, one needs to test the planarity of $\D'_i$ for every $i \in [k]$. 
It is well known that this task can be done in time linear in the 
size of $\D$~\cite{HopcroftTarjan1974}. 
Thus, the Recognition Phase can be done in polynomial time. 

\item {\bf Min-Rank Computation Phase:}\\
Upon the completion of the Recognition Phase, if it is confirmed that $\LL$ is indeed the line digraph
of a partially planar digraph, then the second phase is executed to compute $\mrq(\LL)$. 
We show that this phase can also be done in polynomial time. 
Indeed, by Lemma~\ref{lem:strongly_connected_component}, it suffices to show that
$\mrq(\LL_i)$ for $i \in [k]$ can be found in polynomial time. 

On the one hand, since $\D'_i$ (which is isomorphic to $\D_i$) 
is planar, as shown by Lucchesi and Younger~\cite{LucchesiYounger1978}, $\nu_1(\D'_i) = \tau_1(\D'_i)$. 
Therefore, by Proposition~\ref{pro:digraph_line-digraph}, 
\[
\mrq(\LL_i) = |\E(\D'_i)| - \nu_1(\D'_i).
\] 
On the other hand, $\nu_1(\D'_i)$ can be computed in polynomial time \cite{Lucchesi1976}.  
Therefore $\mrq(\LL_i)$ for each $i \in [k]$ can be computed in polynomial time. 
Thus, $\mrq(\LL)$ can be found in polynomial time. \qedhere
\end{enumerate}
\end{proof} 
\vskip 10pt 

In summary, we have the following. 

\vskip 10pt 
\begin{corollary} 
There are polynomial time algorithms to recognize a member and subsequently determine the min-rank of that member of the following families of digraphs: connectively reducible digraphs (which includes fully reducible flow digraphs and cyclically reducible digraphs), and line digraphs of partially planar digraphs. 
\end{corollary}

\section{Conclusion and Open Problems}
\label{sec:conclusion}

We have characterized the ICSI instances whose
optimal scalar linear index codes have near-extreme 
transmission rates. 
Except for one case, these ICSI instances are also those 
that have near-extreme \emph{vector nonlinear} transmission rates. 
We have also introduced an upper bound on min-ranks of
digraphs. Based on this bound, we have discovered several new
families of digraphs whose min-ranks can be found in polynomial time. 

We state below a couple of interesting open problems
for future research.\\
 
\nin {\bf Open Problem I:} Examine the hardness of the decision problem whether a given digraph has min-rank two over a \emph{nonbinary} field $\fq$. \\

\nin {\bf Open Problem II:} Examine the hardness of the problem 
of finding $\bt_q(\D)$ for a given digraph $\D$. \\

\nin {\bf Open Problem III:} Find new families of digraphs whose 
min-ranks can be found in polynomial time.

\section{Acknowledgment}
The authors wish to thank M. Langberg for providing the preprints 
\cite{HavivLangberg2011, BerlinerLangberg2011}. 

\bibliographystyle{IEEEtran}
\bibliography{IndexCodes_of_ExtremeRates}

\section{Appendix}
\begin{proof}[Proof of Theorem~\ref{thm:mr_n-2}]
For the ONLY IF direction, suppose that $\mrq(\G) = n-2$. 
By the maximum-matching bound, $n - 2 \leq n - \mG$. 
Hence $\mG \leq 2$. As $\mG \in \{0,1\}$ and $|\VG| \geq 6$ imply that
either $\G$ has no edges ($\mrq(\G)=n > n -2$) or $\G$ is a star graph 
($\mrq(\G)=n-1 > n-2$), we deduce that $\mG = 2$. 
Moreover, as the graph $F$ has min-rank
\emph{three} less than its order, $\G$ should not contain any subgraph
isomorphic to $F$. Indeed, suppose for otherwise that $F'$ is a subgraph of 
$\G$ and $\F'$ is isomorphic to $F$. 

Consider the following block diagonal matrix $\bM$ with two blocks $\bB_1$ 
and $\bB_2$. 
The first block $\bB_1$, a $6 \times 6$ matrix, corresponds to the rows and 
columns labeled by the vertices in $F'$. 
Moreover, we choose $\bB_1$ so that it has $q$-rank three. 
This is possible since $F'$ is isomorphic to $F$ and $\mrq(F) = 3$.
(Note that $3 = \alpha(F) \leq \mrq(F) \leq \cc(F) = 3$ implies that $\mrq(F) = 3$.)
The second block $\bB_2$ is chosen to be an $(n - 6) \times (n-6)$ identity matrix.  
It corresponds to the rows and columns labeled by the vertices in $\VG \setminus \V(F')$.
Then $\bM$ fits $\G$ and moreover, 
\[
\begin{split}
\rank(\bM) &= \rank(\bB_1) + \rank(\bB_2)\\
&= 3 + (n - 6)\\
&= n - 3.
\end{split}
\]
This implies that $\mrq(\G) \leq n - 3 < n - 2$, which is impossible. 
    
We now turn to the IF direction. 
Suppose that $\mG = 2$ and $\G$ does not contain any subgraph isomorphic
to $F$. Then by the maximum-matching bound, 
$\mrq(\G) \leq n - 2$. As $\aG \leq \mrq(\G)$, it suffices to show that $\aG = n - 2$. 

Let $\{a,b\}$ and $\{c,d\}$ be the two edges of a maximum matching $M$ in $\G$.
Let $U = \{a,b,c,d\}$ and $V = \VG \setminus U$. 
As $\G$ has at least six vertices, suppose that $V = \{f,g,\ldots\}$, where $f \neq g$. 
Since $M$ is a maximum matching, $V$ must be an independent set in $\G$.
The idea is to show that we can always find two nonadjacent vertices in $U$ that are not 
adjacent to any vertex in $V$. Such two vertices can be added to $V$ to 
obtain an independent set of size $n-2$, which establishes the proof.
We refer to such a pair of vertices as an \emph{independent pair}. 

For disjoint subsets $I$ and $J$ of $\VG$, let 
\[
s_\G(I, J) = \big| \big\{ \{i,j\}: \ i \in I, \ j \in J, \ \{i,j\} \in \EG\big\} \big|.
\] 
Based on how the vertices in $U$ are connected to each other, we consider the 
following five cases. Note that we only consider non-isomorphic configurations.\\
 
\nin {\bf Case 1:} $s_\G(\{a,b\},\{c,d\}) = 0$. 
\begin{figure}[h]
\centering
\scalebox{1} 
{
\begin{pspicture}(0,-1.2129687)(3.4828124,1.2129687)
\psline[linewidth=0.04cm,dotsize=0.07055555cm 2.0]{**-**}(0.1809375,0.69453126)(1.3409375,0.71453124)
\psline[linewidth=0.04cm,dotsize=0.07055555cm 2.0]{**-**}(2.0609374,0.71453124)(3.2609375,0.71453124)
\psline[linewidth=0.04cm,linestyle=dashed,dash=0.16cm 0.16cm,dotsize=0.07055555cm 2.0]{-**}(0.2209375,0.69453126)(1.2809376,-0.80546874)
\psline[linewidth=0.04cm,linestyle=dashed,dash=0.16cm 0.16cm,dotsize=0.07055555cm 2.0]{-**}(2.1209376,0.71453124)(2.1609375,-0.8254688)
\psline[linewidth=0.04cm,linestyle=dashed,dash=0.16cm 0.16cm](2.1009376,0.7345312)(1.2409375,-0.7054688)
\psline[linewidth=0.04cm,linestyle=dashed,dash=0.16cm 0.16cm](1.2609375,0.69453126)(1.2409375,-0.72546875)
\usefont{T1}{ptm}{m}{n}
\rput(0.22234374,1.0045313){$a$}
\usefont{T1}{ptm}{m}{n}
\rput(1.2723438,1.0245312){$b$}
\usefont{T1}{ptm}{m}{n}
\rput(2.1023438,1.0245312){$c$}
\usefont{T1}{ptm}{m}{n}
\rput(3.1723437,1.0245312){$d$}
\usefont{T1}{ptm}{m}{n}
\rput(1.2023437,-1.0354687){$f$}
\usefont{T1}{ptm}{m}{n}
\rput(2.1023438,-1.0354687){$g$}
\end{pspicture} 
}
\caption{Case 1}
\label{fig:conf1}
\end{figure}

There are four candidates for an independent pair, namely
$\{a,c\}$, $\{a,d\}$, $\{b,c\}$, $\{b,d\}$. All of these pairs
fail to be an independent pair if and only if either both $a$ 
and $b$ are adjacent to some vertices in $V$ or both $c$ 
and $d$ are adjacent to some vertices in $V$. We show that
either case never happens, by contradiction. 

Suppose both $a$ and $b$ are adjacent to some vertices in $V$. 
(The case when both $c$ and $d$ are adjacent to some vertices in $V$
is investigated analogously.)  
Without loss of generality, assume that $a$ and $f$ are adjacent. 
Then $b$ must be adjacent to $f$ but not to any other vertex in $\V$. 
Indeed, if $b$ is adjacent to $h \in V$, $h \neq f$, 
then the set of three edges $\{a,f\}$, $\{b,h\}$, and $\{c,d\}$
form a matching of size three, which is impossible since $\mG = 2$.
Similarly, $a$ should not be adjacent to any other vertex in $V$ rather
than $f$. 

As $\G$ is connected, $f$ must be adjacent to either $c$ or $d$.
Without loss of generality, suppose $f$ and $c$ are adjacent. On the one hand, since $\G$
is connected, $g$ must be adjacent to some vertex in $U$. 
On the other hand, $g$ cannot be adjacent to any vertex in $U$, as
\begin{itemize}
	\item if $g$ and $a$ are adjacent, then $\{a,g\}$, $\{b,f\}$, and 
	$\{c,d\}$ form a matching of size three, which is impossible; 
	\item if $g$ and $b$ are adjacent, then $\{a,f\}$, $\{b,g\}$, and 
	$\{c,d\}$ form a matching of size three, which is impossible;
	\item if $g$ and $c$ are adjacent, then $\G$ has a subgraph isomorphic to $F$ (see Fig.~\ref{fig:conf1}), which is impossible; 
	\item if $g$ and $d$ are adjacent, then $\{a,b\}$, $\{c,f\}$, and 
	$\{d,g\}$ form a matching of size three, which is impossible.  
\end{itemize}
We obtain a contradiction. \\

\nin {\bf Case 2:} $s_\G(\{a,b\},\{c,d\}) = 1$. 
Without loss of generality, suppose that $\{b,c\}$ is the only edge that connects
$\{a,b\}$ and $\{c,d\}$. 

There are three candidates for an independent pair, namely
$\{a,c\}$, $\{a,d\}$, and $\{b,d\}$. All of these three pairs 
fail to be an independent pair only if at least one of the pairs
$\{a,b\}$, $\{a,d\}$, and $\{c,d\}$ has both vertices adjacent
to some vertices in $V$. We show below that this scenario 
cannot happen. 

\begin{enumerate}
	\item Assume that both $a$ and $b$ are adjacent to some vertices in $V$. 
	\begin{figure}[h]
	\centering
	\scalebox{1} 
{
\begin{pspicture}(0,-1.2129687)(3.4828124,1.2129687)
\psline[linewidth=0.04cm,dotsize=0.07055555cm 2.0]{**-**}(0.1809375,0.69453126)(1.3409375,0.71453124)
\psline[linewidth=0.04cm,dotsize=0.07055555cm 2.0]{**-**}(2.0609374,0.71453124)(3.2609375,0.71453124)
\psline[linewidth=0.04cm,linestyle=dashed,dash=0.16cm 0.16cm,dotsize=0.07055555cm 2.0]{-**}(0.2209375,0.69453126)(1.2809376,-0.80546874)
\psline[linewidth=0.04cm,linestyle=dashed,dash=0.16cm 0.16cm,dotsize=0.07055555cm 2.0]{-**}(2.1209376,0.71453124)(2.1609375,-0.8254688)
\psline[linewidth=0.04cm,linestyle=dashed,dash=0.16cm 0.16cm](1.2609375,0.69453126)(1.2409375,-0.72546875)
\usefont{T1}{ptm}{m}{n}
\rput(0.22234374,1.0045313){$a$}
\usefont{T1}{ptm}{m}{n}
\rput(1.2723438,1.0245312){$b$}
\usefont{T1}{ptm}{m}{n}
\rput(2.1023438,1.0245312){$c$}
\usefont{T1}{ptm}{m}{n}
\rput(3.1723437,1.0245312){$d$}
\usefont{T1}{ptm}{m}{n}
\rput(1.2023437,-1.0354687){$f$}
\usefont{T1}{ptm}{m}{n}
\rput(2.1023438,-1.0354687){$g$}
\psline[linewidth=0.04cm](1.2809376,0.71453124)(2.1409376,0.71453124)
\end{pspicture} 
}
\caption{Sub-case 1}
\label{fig:subcase1}
	\end{figure}
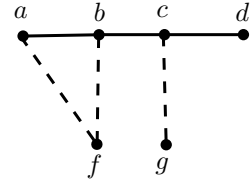
Suppose without loss of generality that $a$ and $f$ are adjacent. 
Then the same argument as in Case 1 establishes that $b$ must be adjacent to $f$ but not to any other vertex in $V$. 
On the one hand, as $\G$ is connected, $g$ must be adjacent to 
some vertex in $U$. 
On the other hand, as $\mG = 2$, $g$ should not be adjacent to any vertex among $a$, $b$, and $d$.
Moreover, $g$ and $c$ cannot be adjacent, for otherwise $\G$ would contain a subgraph isomorphic
to $F$ (see Fig.~\ref{fig:subcase1}).   
We obtain a contradiction.  	
	
\item Assume that both $a$ and $d$ are adjacent to some vertices in $V$. 
	\begin{figure}[h]
	\centering	
	\scalebox{1} 
{
\begin{pspicture}(0,-1.2129687)(3.4828124,1.2129687)
\psline[linewidth=0.04cm,dotsize=0.07055555cm 2.0]{**-**}(0.1809375,0.69453126)(1.3409375,0.71453124)
\psline[linewidth=0.04cm,dotsize=0.07055555cm 2.0]{**-**}(2.0609374,0.71453124)(3.2609375,0.71453124)
\psline[linewidth=0.04cm,linestyle=dashed,dash=0.16cm 0.16cm,dotsize=0.07055555cm 2.0]{-**}(0.2209375,0.69453126)(1.2809376,-0.80546874)
\usefont{T1}{ptm}{m}{n}
\rput(0.22234374,1.0045313){$a$}
\usefont{T1}{ptm}{m}{n}
\rput(1.2723438,1.0245312){$b$}
\usefont{T1}{ptm}{m}{n}
\rput(2.1023438,1.0245312){$c$}
\usefont{T1}{ptm}{m}{n}
\rput(3.1723437,1.0245312){$d$}
\usefont{T1}{ptm}{m}{n}
\rput(1.2023437,-1.0354687){$f$}
\usefont{T1}{ptm}{m}{n}
\rput(2.1023438,-1.0354687){$g$}
\psline[linewidth=0.04cm](1.2809376,0.71453124)(2.1409376,0.71453124)
\psdots[dotsize=0.152](2.1609375,-0.7654688)
\psline[linewidth=0.04cm,linestyle=dashed,dash=0.16cm 0.16cm](3.1809375,0.71453124)(1.2809376,-0.68546873)
\end{pspicture} 
}
\caption{Sub-case 2}
\label{fig:subcase2}
 \end{figure}
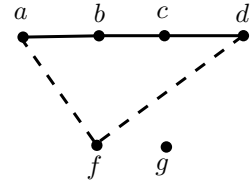
Suppose without loss of generality that $a$ and $f$ are adjacent. 
As there are no matchings of size three in $\G$, $d$ is adjacent to $f$ but not to any other
vertex in $V$. Also, 
$g$ is not adjacent to any vertex in $U$. However, this would 
imply that $g$ is an isolated vertex of $\G$, which is impossible
as $\G$ is connected.

\item Assume that both $c$ and $d$ are adjacent to some vertices in $V$. 
This sub-case is completely similar to the first sub-case. 
\end{enumerate}

\nin {\bf Case 3:} $s_\G(\{a,b\},\{c,d\}) = 2$ and the two edges that connect $\{a,b\}$ and $\{c,d\}$ share one common vertex. Without loss of 
generality suppose that these two edges are $\{b,c\}$ and $\{b,d\}$. 

There are two candidates for an independent pair, namely
$\{a,c\}$ and $\{a,d\}$. It suffices to show that $a$ is not adjacent to any
vertex in $V$ and either $c$ or $d$ is not adjacent to any vertex
in $V$. 

Suppose that $a$ is adjacent to a vertex, say $f$, in $V$. 
\begin{figure}[h]
\centering
\scalebox{1} 
{
\begin{pspicture}(0,-1.47375)(3.7028127,1.45375)
\psline[linewidth=0.04cm,dotsize=0.07055555cm 2.0]{**-**}(0.19953126,0.43375)(1.3595313,0.45375)
\psline[linewidth=0.04cm,dotsize=0.07055555cm 2.0]{**-**}(2.0795312,0.45375)(3.2795312,0.45375)
\usefont{T1}{ptm}{m}{n}
\rput(0.22234374,0.74375){$a$}
\usefont{T1}{ptm}{m}{n}
\rput(1.0723437,0.76375){$b$}
\usefont{T1}{ptm}{m}{n}
\rput(2.1023438,0.76375){$c$}
\usefont{T1}{ptm}{m}{n}
\rput(3.3923438,0.76375){$d$}
\usefont{T1}{ptm}{m}{n}
\rput(1.2023437,-1.29625){$f$}
\usefont{T1}{ptm}{m}{n}
\rput(2.1023438,-1.29625){$g$}
\psline[linewidth=0.04cm](1.2995313,0.45375)(2.1595314,0.45375)
\psdots[dotsize=0.152](2.1795313,-1.02625)
\psdots[dotsize=0.152](1.2795312,-1.02625)
\psarc[linewidth=0.04](2.2495313,0.48375){0.95}{0.0}{180.0}
\psline[linewidth=0.04cm,linestyle=dashed,dash=0.16cm 0.16cm](2.1595314,-1.02625)(2.1595314,-1.08625)
\psline[linewidth=0.04cm,linestyle=dashed,dash=0.16cm 0.16cm](0.26,0.43375)(1.26,-0.98625)
\psline[linewidth=0.04cm,linestyle=dashed,dash=0.16cm 0.16cm](0.28,0.45375)(2.18,-1.02625)
\end{pspicture} 
}
\caption{}
\label{fig:a}
\end{figure} 
As $\mG = 2$, we deduce that $g$ is not adjacent to any vertex among $b$, $c$, and $d$. Also, since $\G$ does not contain a 
subgraph isomorphic to $F$, we deduce that $g$ cannot be adjacent to $a$ (see Fig.~\ref{fig:a}). Hence $g$ is an isolated vertex of $\G$, which
is impossible as $\G$ is connected. 

Now suppose that both $c$ and $d$ are adjacent to some vertices
in $V$. Without loss of generality, suppose that $c$ is adjacent to 
$f$. 
\begin{figure}[h]
\centering
\scalebox{1} 
{
\begin{pspicture}(0,-1.47375)(3.7028124,1.45375)
\psline[linewidth=0.04cm,dotsize=0.07055555cm 2.0]{**-**}(0.1809375,0.43375)(1.3409375,0.45375)
\psline[linewidth=0.04cm,dotsize=0.07055555cm 2.0]{**-**}(2.0609374,0.45375)(3.2609375,0.45375)
\usefont{T1}{ptm}{m}{n}
\rput(0.22234374,0.74375){$a$}
\usefont{T1}{ptm}{m}{n}
\rput(1.0723437,0.76375){$b$}
\usefont{T1}{ptm}{m}{n}
\rput(2.1023438,0.76375){$c$}
\usefont{T1}{ptm}{m}{n}
\rput(3.3923438,0.76375){$d$}
\usefont{T1}{ptm}{m}{n}
\rput(1.2023437,-1.29625){$f$}
\usefont{T1}{ptm}{m}{n}
\rput(2.1023438,-1.29625){$g$}
\psline[linewidth=0.04cm](1.2809376,0.45375)(2.1409376,0.45375)
\psdots[dotsize=0.152](2.1609375,-1.02625)
\psdots[dotsize=0.152](1.2609375,-1.02625)
\psarc[linewidth=0.04](2.2309375,0.48375){0.95}{0.0}{180.0}
\psline[linewidth=0.04cm,linestyle=dashed,dash=0.16cm 0.16cm](2.1209376,0.45375)(1.2809376,-0.96625)
\psline[linewidth=0.04cm,linestyle=dashed,dash=0.16cm 0.16cm](3.1609375,0.43375)(1.2809376,-1.00625)
\psline[linewidth=0.04cm,linestyle=dashed,dash=0.16cm 0.16cm](1.3209375,0.37375)(2.1609375,-1.08625)
\psline[linewidth=0.04cm,linestyle=dashed,dash=0.16cm 0.16cm](2.1409376,-1.02625)(2.1409376,-1.08625)
\end{pspicture} 
}
\caption{}
\label{fig:cd}
\end{figure} 
Then since $\mG = 2$, $d$ must be adjacent to $f$ but not to any other vertex in $V$. Also, $g$ cannot be adjacent to any vertex 
among $a$, $c$, and $d$ for the same reason. Moreover, as $\G$
does not contain a subgraph isomorphic to $F$, we deduce that
$g$ is not adjacent to $b$ (see Fig.~\ref{fig:cd}). 
(Indeed, if $g$ and $b$ are adjacent, then the following subgraph of
$\G$ is isomorphic to $F$: 
its vertex set is $\{a,b,c,d,f,g\}$, and its edge set is 
$\big\{\{c,d\}, \{d,f\}, \{c,f\}, \{c,b\}, \{b,a\}, \{b,g\} \big\}$.)
Therefore, $g$ is an isolated vertex of $\G$. We obtain a contradiction.\\
 
\nin {\bf Case 4:} $s_\G(\{a,b\},\{c,d\}) = 2$ and the two edges that connect $\{a,b\}$ and $\{c,d\}$ share no common vertices. 
Suppose, without loss of generality, that these two edges are $\{a,d\}$ and $\{b,c\}$.  

There are two candidates for an independent pair, namely
$\{a,c\}$ and $\{b,d\}$. Both of these pairs fail to be an 
independent pair if and only if at least one of the four pairs
$\{a,b\}$, $\{a,d\}$, $\{c,b\}$, and $\{c,d\}$ has both vertices
adjacent to some vertices in $V$. 
By symmetry, it suffices to show that the scenario when both
$a$ and $b$ are adjacent to some vertices in $V$ never happens.

Suppose now that $a$ and $b$ are adjacent to some vertices in $V$. 

\begin{figure}[H]
\centering
\scalebox{1} 
{
\begin{pspicture}(0,-1.0740604)(4.429995,1.7946434)
\psline[linewidth=0.04cm,dotsize=0.07055555cm 2.0]{**-**}(0.92671394,0.9237521)(2.086714,0.9437521)
\psline[linewidth=0.04cm,dotsize=0.07055555cm 2.0]{**-**}(2.8067138,0.9437521)(4.0067143,0.9437521)
\usefont{T1}{ptm}{m}{n}
\rput(0.8295265,1.253752){$a$}
\usefont{T1}{ptm}{m}{n}
\rput(1.9395264,1.253752){$b$}
\usefont{T1}{ptm}{m}{n}
\rput(2.8295264,1.253752){$c$}
\usefont{T1}{ptm}{m}{n}
\rput(4.1195264,1.253752){$d$}
\usefont{T1}{ptm}{m}{n}
\rput(1.9495264,-0.8662479){$f$}
\usefont{T1}{ptm}{m}{n}
\rput(2.8495266,-0.8462479){$g$}
\psline[linewidth=0.04cm](2.0267138,0.9437521)(2.8867137,0.9437521)
\psdots[dotsize=0.152](2.9067137,-0.5362479)
\psdots[dotsize=0.152](2.006714,-0.5362479)
\rput{-33.111652}(0.3979387,1.3386276){\psarc[linewidth=0.04](2.450477,0.0){1.7707293}{66.985695}{180.0}}
\psline[linewidth=0.04cm,linestyle=dashed,dash=0.16cm 0.16cm](2.8867137,-0.5362479)(2.8867137,-0.5962479)
\psline[linewidth=0.04cm,linestyle=dashed,dash=0.16cm 0.16cm](1.026714,0.9037521)(1.986714,-0.47624794)
\psline[linewidth=0.04cm,linestyle=dashed,dash=0.16cm 0.16cm](2.006714,0.9237521)(2.006714,-0.51624787)
\end{pspicture} 
}\caption{Case 4}
\label{fig:case4}
\end{figure}  
Suppose that $a$ and $f$ are adjacent. 
The condition that $\mG = 2$ forces $b$ to be adjacent to $f$ but not 
to any other vertex in $V$. That condition also implies that $g$
must be an isolated vertex in $\G$, which is impossible as $\G$ is
connected.\\

\nin {\bf Case 5:} $s_\G(\{a,b\},\{c,d\}) = 3$. Without loss of generality, suppose that $\{a,d\}$, $\{b,c\}$, and $\{b,d\}$ are 
the edges that connect $\{a,b\}$ and $\{c,d\}$. The only candidate
for an independent pair is $\{a,c\}$. We prove by contradiction that
both $a$ and $c$ are not adjacent to any vertex in $V$. By symmetry,
it suffices to verify this property for only one of them. 

Suppose that $a$ is adjacent to some vertex in $V$. Let $a$ be 
adjacent to $f$. 

\begin{figure}[h]
\centering
\scalebox{1} 
{
\begin{pspicture}(0,-1.4873681)(4.156894,1.4678612)
\psline[linewidth=0.04cm,dotsize=0.07055555cm 2.0]{**-**}(0.67220664,0.4904443)(1.8322066,0.5104443)
\psline[linewidth=0.04cm,dotsize=0.07055555cm 2.0]{**-**}(2.5522065,0.5104443)(3.7522066,0.5104443)
\usefont{T1}{ptm}{m}{n}
\rput(0.67642534,0.8004443){$a$}
\usefont{T1}{ptm}{m}{n}
\rput(1.5264252,0.8204443){$b$}
\usefont{T1}{ptm}{m}{n}
\rput(2.5564253,0.8204443){$c$}
\usefont{T1}{ptm}{m}{n}
\rput(3.8464253,0.8204443){$d$}
\usefont{T1}{ptm}{m}{n}
\rput(1.7164253,-1.2795556){$f$}
\usefont{T1}{ptm}{m}{n}
\rput(2.6364255,-1.2595556){$g$}
\psline[linewidth=0.04cm](1.7722068,0.5104443)(2.6322067,0.5104443)
\psdots[dotsize=0.152](2.6522067,-0.9695557)
\psdots[dotsize=0.152](1.7522067,-0.9695557)
\rput{-27.002811}(0.331856,0.98395854){\psarc[linewidth=0.04](2.2149477,-0.19908594){1.6467787}{54.036804}{180.0}}
\psline[linewidth=0.04cm,linestyle=dashed,dash=0.16cm 0.16cm](2.6322067,-0.9695557)(2.6322067,-1.0295557)
\psline[linewidth=0.04cm,linestyle=dashed,dash=0.16cm 0.16cm](0.7326753,0.4904443)(1.7326753,-0.9295557)
\psline[linewidth=0.04cm,linestyle=dashed,dash=0.16cm 0.16cm](0.75267535,0.5104443)(2.6526754,-0.9695557)
\rput{-207.20683}(5.5443707,0.5693966){\psarc[linewidth=0.04](2.7032914,0.9555354){1.0831244}{54.036804}{180.0}}
\end{pspicture} 
}
\caption{Case 5}
\label{fig:case5}
\end{figure}
As $\mG = 2$ and $\G$ is connected, $g$ must
be adjacent to $a$. However, $\G$ now contains a subgraph whose edge set
consists of $\{b,c\},\{b,d\},\{c,d\},\{b,a\},\{a,f\},\{a,g\},$ which is isomorphic to $F$ (see Fig.~\ref{fig:case5}). This contradicts our assumption. \\

\nin {\bf Case 6:} $s_\G(\{a,b\},\{c,d\}) = 4$. In this case, the subgraph of 
$\G$ induced by $\{a,b,c,d\}$ is a complete graph. 

\begin{figure}[h]
\centering
\scalebox{1} 
{
\begin{pspicture}(0,-1.5073681)(4.1383004,1.4678612)
\psline[linewidth=0.04cm,dotsize=0.07055555cm 2.0]{**-**}(0.6722067,0.49044433)(1.8322067,0.51044434)
\psline[linewidth=0.04cm,dotsize=0.07055555cm 2.0]{**-**}(2.5522065,0.51044434)(3.7522066,0.51044434)
\usefont{T1}{ptm}{m}{n}
\rput(0.65783167,0.8004443){$a$}
\usefont{T1}{ptm}{m}{n}
\rput(1.5078316,0.82044435){$b$}
\usefont{T1}{ptm}{m}{n}
\rput(2.5378315,0.82044435){$c$}
\usefont{T1}{ptm}{m}{n}
\rput(3.8278315,0.82044435){$d$}
\usefont{T1}{ptm}{m}{n}
\rput(1.6978315,-1.3195555){$f$}
\usefont{T1}{ptm}{m}{n}
\rput(2.6178317,-1.2795556){$g$}
\psline[linewidth=0.04cm](1.7722069,0.51044434)(2.6322067,0.51044434)
\psdots[dotsize=0.152](2.6522067,-0.9695557)
\psdots[dotsize=0.152](1.7522068,-0.9695557)
\rput{-27.002811}(0.331856,0.98395854){\psarc[linewidth=0.04](2.2149477,-0.1990859){1.6467787}{54.036804}{180.0}}
\psline[linewidth=0.04cm,linestyle=dashed,dash=0.16cm 0.16cm](2.6322067,-0.9695557)(2.6322067,-1.0295557)
\rput{-207.20683}(5.544371,0.5693964){\psarc[linewidth=0.04](2.703292,0.95553535){1.0831244}{54.036804}{180.0}}
\rput{-207.20683}(3.627337,1.0126611){\psarc[linewidth=0.04](1.6911422,0.94521755){1.0592493}{54.036804}{180.0}}
\end{pspicture} 
}
\caption{Case 6}
\end{figure}

As $\G$ is connected, both $f$ and $g$ must be adjacent to 
some vertices in $U$. 
If $f$ and $g$ are adjacent to the same vertex in $U$, 
then $\G$ contains a subgraph isomorphic to $F$, which contradicts
our assumption. For instance, if both $f$ and $g$ are adjacent to $a$, 
then this subgraph has vertex set $\{a,b,c,d,f,g\}$ 
and edge set consisting of the edges $\{b,c\}, \{c,d\}, \{b,d\}, \{b,a\}, \{a,f\}, \{a,g\}$.
It is also easy to verify that if $f$ and $g$ are adjacent to different vertices in $U$, then 
$\G$ contains a matching of size three. This contradicts our assumption that
$\mG = 2$.   
Thus, Case 6 never happens. 
\end{proof} 

\end{document}